\documentclass{LMCS}

\def\dOi{11(1:4)2015}
\lmcsheading%
{\dOi}
{1--23}
{}
{}
{Nov.~\phantom01, 2012}
{Mar.~\phantom04, 2015}
{}

\ACMCCS{[{\bf Theory of computation}]: Logic---Proof theory}

\usepackage[lutzsyntax]{virginialake}\vlsmallbrackets\aftrianglefalse
\usepackage{graphicx}
\usepackage{stmaryrd}
\usepackage{hyperref}
\usepackage[anythingbreaks]{breakurl}
\usepackage{turnstile}
\usepackage{latexsym}
\usepackage{amsmath}
\usepackage{ifpdf}
\usepackage{latexsym,multirow,multicol}\allowdisplaybreaks[2]
\usepackage{mathdots}
\usepackage{mathrsfs}
\usepackage{rotating}

\theoremstyle{definition}

\newcommand{\SKS}{\mathsf{SKS}}
\newcommand{\KS}{\mathsf{KS}}

\newcommand{\fff}{\bot}
\newcommand{\ttt}{\top}

\newcommand{\aSKS}{\mathsf{KS}^+}

\newcommand{\GNorm}{\mathsf{norm}}

\newcommand{\ai  }{\mathsf{ai}}
\newcommand{\aw  }{\mathsf{aw}}
\newcommand{\ac  }{\mathsf{ac}}
\newcommand{\aid }{{\ai{\downarrow}}}
\newcommand{\awd }{{\aw{\downarrow}}}
\newcommand{\acd }{{\ac{\downarrow}}}
\newcommand{\aiu }{{\ai{\uparrow}}}
\newcommand{\awu }{{\aw{\uparrow}}}
\newcommand{\acu }{{\ac{\uparrow}}}
\newcommand{\swi }{\mathsf{s}}
\newcommand{\med }{\mathsf{m}}

\newcommand  {\gw  }{\mathsf w}
\newcommand  {\gwd }{{\gw{\downarrow}}}
\newcommand  {\gwu }{{\gw{\uparrow}}}
\newcommand  {\gc  }{\mathsf c}
\newcommand  {\gcdown }{{\gc{\downarrow}}}
\newcommand  {\gcu }{{\gc{\uparrow}}}
\newcommand  {\gi  }{\mathsf i}
\newcommand  {\giu }{{\gi{\uparrow}}}
\newcommand  {\gidown }{{\gi{\downarrow}}}

\newcommand{\res}{\mathsf{res}}

\newcommand{\one}{}
\newcommand{\two}{}

\newcommand{\rwdcd}{{{\mathsf w}{\downarrow}{\hbox{-}}{\mathsf c}{\downarrow}}}

\newcommand{\rwdwu}{{{\mathsf w}{\downarrow}{\hbox{-}}{\mathsf w}{\uparrow  }}}
\newcommand{\rwdcu}{{{\mathsf w}{\downarrow}{\hbox{-}}{\mathsf c}{\uparrow  }}}
\newcommand{\rcuwu}{{{\mathsf c}{\uparrow  }{\hbox{-}}{\mathsf w}{\uparrow  }}}
\newcommand{\rcdwu}{{{\mathsf c}{\downarrow}{\hbox{-}}{\mathsf w}{\uparrow  }}}

\newcommand{\rcdcu}{{{\mathsf c}{\downarrow}{\hbox{-}}{\mathsf c}{\uparrow  }}}
\newcommand{\ridwu}{{{\mathsf i}{\downarrow}{\hbox{-}}{\mathsf w}{\uparrow  }}}
\newcommand{\ridcu}{{{\mathsf i}{\downarrow}{\hbox{-}}{\mathsf c}{\uparrow  }}}

\newcommand{\RD}[1]{{\color{Red}#1}}
\newcommand{\GR}[1]{{\color{Green}#1}}
\newcommand{\DO}[1]{{\color{DarkOrchid}#1}}
\newcommand{\PB}[1]{{\color{ProcessBlue}#1}}
\newcommand{\MG}[1]{{\color{Magenta}#1}}
\newcommand{\SG}[1]{{\color{SpringGreen}#1}}
\newcommand{\RS}[1]{{\color{RawSienna}#1}}
\newcommand{\YO}[1]{{\color{YellowOrange}#1}}
\newcommand{\PW}[1]{{\color{Periwinkle}#1}}

\newcommand{\three}{}
\newcommand{\four }{}

\newcommand{\flow}{fl}

\begin{document}
\title[On The Relative Proof Complexity of Deep Inference via Atomic Flows]{On the Relative Proof Complexity of Deep Inference via Atomic Flows}
\author{Anupam Das}
\address{\'Ecole Normale Sup\'erieure de Lyon (ENS Lyon), France}
\email{anupam.das@ens-lyon.fr}

\begin{abstract}
We consider the proof complexity of the minimal complete fragment, $\KS$, of standard deep inference systems for propositional logic. To examine the size of proofs we employ \emph{atomic flows}, diagrams that trace structural changes through a proof but ignore logical information. As results we obtain a polynomial simulation of versions of Resolution, along with some extensions. We also show that these systems, as well as bounded-depth Frege systems, cannot polynomially simulate $\KS$, by giving polynomial-size proofs of certain variants of the propositional pigeonhole principle in $\KS$.
\end{abstract}

\maketitle

\section{Introduction}
Deep inference is a relatively recent proof methodology whose systems differ from other formalisms by allowing derivations themselves to be composed by logical connectives. One of its main features is \emph{locality}, i.e.\ inference steps can be checked in constant time, a property that is impossible to achieve in Gentzen systems \cite{Brun:03:Two-Rest:mn}. In recent years there has been an increasing interest in the proof complexity of deep inference \cite{BrusGugl:07:On-the-P:fk} \cite{Jera::On-the-C:kx} \cite{Stra:08:Extensio:kk} \cite{BrusGuglGundPari:09:Quasipol:kx} \cite{Das:11:Depth-Change} \cite{Das:13:The-Pige:fk}, in particular the weaker systems initially introduced by Br\"{u}nnler and Tiu \cite{BrunTiu:01:A-Local-:ly}. Perhaps the most notable result is that a certain system, denoted $\KS\cup\{\gcu\}$, quasipolynomially\footnote{A quasipolynomial in $n$ is a function $2^{\log^{k} n }$.} simulates Frege systems \cite{Jera::On-the-C:kx} \cite{BrusGuglGundPari:09:Quasipol:kx}. It is conjectured that this can be improved to a polynomial simulation, and so proving nontrivial lower bounds for $\KS\cup\{\gcu\}$ is likely equivalent to proving them for Frege systems, a task which has escaped proof complexity theorists for years.

However this quasipolynomial simulation relies crucially on the presence of ``dag-like behaviour'', manifested in deep inference by a particular rule, \emph{cocontraction}: $\vlinf{\gcu}{}{A\vlan A}{A}$. Without it we have a minimal complete system closed under deep inference, $\KS$. This system is free of compression mechanisms, in the sense that a proof of a conjunction can be `partitioned' into proofs of each conjunct, unlike proofs in systems that are dag-like or contain cut. This is explained further in \cite{Das:12:CompMech}.

It is conjectured that $\KS$ is unable to polynomially simulate $\KS\cup\{\gcu\}$ \cite{BrusGugl:07:On-the-P:fk} \cite{BrusGuglGundPari:09:Quasipol:kx} \cite{Das:11:Depth-Change}  \cite{Stra:08:Extensio:kk}, raising the question of where exactly it fits in the hierarchy of proof systems.

\emph{Atomic flows} are diagrams that track structural changes in a proof (duplication, creation and destruction of atoms) but ignore logical information. In this work they serve as a useful abstraction because of certain rewriting procedures on them which can be used to manipulate derivations soundly without any mention of logical syntax. Atomic flows are introduced formally in \cite{GuglGund:07:Normalis:lr} and a comprehensive account can be found in \cite{Gund:09:A-Genera:kx}.

In this paper we focus on upper bounds and simulations to demonstrate the relative strength of $\KS$. The starting points in our arguments are proofs in a system obtained by extending $\KS\cup\{\gcu\}$ by the \emph{coweakening} rule: $\vlinf{\gwu}{}{\ttt}{A}$; the resulting system $\KS\cup\{\gcu,\gwu\}$ is denoted $\aSKS$ in this paper. We then appeal to sound rewriting rules on the atomic flows of these proofs to show that cocontraction and coweakening steps can, in certain cases, be eliminated from a proof in polynomial time. 

It is worth mentioning here that the addition of coweakening makes little difference to proof complexity, indeed it is not difficult to see from the rewriting rules in Fig.\ \ref{FigRed} that coweakening steps can be eliminated in time linear in the size of the proof. Rather the real generators of complexity in our proofs are the interactions between contraction and cocontraction nodes in the atomic flows, as we show in Prop.~\ref{CocConSizePath} and Lemma~\ref{NoConLoopsQuad}.

In Sect.~\ref{DITreeTT} we give a simple example of how atomic flows can be used to normalise a na\"{i}ve	 encoding of truth table proofs in $\aSKS$ to produce a polynomial simulation in $\KS$. As a corollary we obtain a superpolynomial separation of $\KS$ from tree-like cut-free Gentzen systems, since they are unable to simulate truth tables \cite{springerlink:10.1007/BF00156916}, a new proof of a result appearing in \cite{BrusGugl:07:On-the-P:fk}.

In Sect.~\ref{StrongViaFPHP} we consider stronger systems; we improve a result of Je{\v r}\'abek's that $\aSKS$ has polynomial-size proofs of the functional and onto versions of the propositional pigeonhole principle by showing that they can be polynomially transformed into $\KS$ proofs of the same conclusions. This immediately entails that cut-free Gentzen systems, Resolution and even bounded-depth Frege systems are exponentially separated from $\KS$ by the results of \cite{PitBeaImp:93:ExpLBPHP} and \cite{KraPudWood:95:ExpPHPbdF}. 

In Sect.~\ref{VerResSim} we consider simulations in $\KS$ of other proof systems. There is already a na\"ive simulation of tree-like cut-free Gentzen sequent calculi, appearing in \cite{BrusGugl:07:On-the-P:fk}, and what amounts to a polynomial simulation of `Resolution with multisets' is outlined in \cite{Gugl:03:Resoluti:ce}. Here we formalise the latter result and also give a polynomial simulation of tree-like Resolution systems, even when sets are the basic data structure. We show that both simulations extend to so-called Resolution$(f)$ systems, introduced by Kraj{\'\i}{\v c}ek in \cite{Krajicek01onthe} and known to be strictly stronger than usual versions of Resolution \cite{Segerlind02aswitching} \cite{Esteban04onthe}.

This paper is a full version of \cite{Das:12:Complexi:kx}, and differs from that work as follows:
\begin{enumerate}
\item Full proofs are given where they were brief or omitted previously.
\item We expand on some of the preliminary work on the complexity of normalisation induced by flow rewriting in Sect.~\ref{prelim}. In particular we provide full proofs of termination and confluence for the rewriting system $\GNorm$ and give explicit reduction strategies that achieve the complexity bounds given in the previous work.
\item In Sect.~\ref{prelim} we use the definitions and notations found in \cite{GuglGund:07:Normalis:lr} and \cite{Gund:09:A-Genera:kx} for atomic flows, to maintain consistency with the existing literature, whereas in \cite{Das:12:Complexi:kx} self-contained definitions were preferred for brevity.
\item There were some errors in the statements of results in the previous work, which have been corrected here. In particular the previously stated simulations of dag-like cut-free Gentzen systems and dag-like Resolution systems are incorrect as presented and we have not been able to amend them. Here we only obtain such simulations for certain cases, namely the versions of Resolution in Sect.~\ref{VerResSim}.
\end{enumerate}

\section{Deep inference, atomic flows and normalisation}\label{prelim}
In this section we introduce deep inference systems for propositional logic and \emph{atomic flows}, diagrams that trace the structural changes in a proof. We consider a graph rewriting system $\GNorm$ on flows, corresponding to sound manipulations on proofs, and analyse the complexity of termination in this system.

\subsection{Deep inference}
We consider propositional logic with formulae constructed from literals (propositional variables and their duals), also called \emph{atoms}, over the basis $\{\ttt, \fff, \wedge,\vee \}$, and use the infix symbol $\equiv$ to denote equivalence of expressions. The variables $a,b,c,d$ range over literals, with $\bar{a}, \bar{b},\dots$ denoting their duals, and $A,B,C,D$ range over formulae; both sets of variables may include subscripts or superscripts as necessary.

For clarity we use square brackets $[,]$ for disjunctions and round ones $(,)$ for conjunctions. We generally omit external brackets of an expression, and also internal ones under associativity. This does not cause any confusion when it comes to proofs, since any valid bracketing can be reduced to any other by the $=$ rule in Dfn.~\ref{RulesSystems}.

Note that we do not have a symbol for negation in our language, formulae are always in negation normal form. We may however write $\bar{A}$ to denote the De Morgan dual of a formula $A$, obtained by the following rules: 
\[
\bar{\fff}\equiv \ttt, \quad \bar{\ttt}\equiv \fff,\quad  \bar{\bar{a}}\equiv a, \quad \overline{A\vlan B} \equiv \bar{A}\vlor \bar{B}, \quad \overline{A\vlor B} \equiv \bar{A}\vlan \bar{B}
\]


\newcommand{\prem}{\mathsf{pr}}
\newcommand{\conc}{\mathsf{cn}}


\begin{defi}[Rules and systems]\label{RulesSystems}
An \emph{inference rule} is a sound binary relation on formulae decidable in polynomial time, and a \emph{system} is a set of rules. We define the rules we use below, and the systems $\KS=\{\aid,\awd,\acd,\swi,\med\}$, $\aSKS=\KS\cup\{\awu,\acu\}$ and $\SKS=\aSKS\cup\{\aiu\}$.

\[
\begin{array}{@{}c@{}c@{}c@{}c@{}}
   \multispan3{\hfil Atomic structural rules\hfil}%
      &\quad\mbox{Linear logical rules}                                        \\
\noalign{\bigskip}

      \vlinf{\aid}{}{\vls[a.{\bar a}]}{\ttt}&
\qquad\vlinf{\awd}{}a\fff&
\qquad\vlinf{\acd}{}a{\vls[a.a]} &
\qquad\vlinf{\swi}{}{\vls[(A.B).C]}{\vls(A.[B.C])}\\
\noalign{\smallskip}
      \emph{identity}&
\qquad\emph{weakening}&
\qquad\emph{contraction}&\qquad\emph{switch} \\
\noalign{\bigskip}
      \vlinf{\aiu}{}\fff{\vls(a.{\bar a})}&
\qquad\vlinf{\awu}{}\ttt a&
\qquad\vlinf{\acu}{}{\vls (a.a)}a  &\qquad
\vlinf{\med}{}{\vls([A.C].[B.D])}
              {\vls[(A.B).(C.D)]}\\
\noalign{\smallskip}
      \emph{cut}&
\qquad\emph{coweakening}&
\qquad\emph{cocontraction}&\qquad\emph{medial}
\end{array}\quad
\]

Note in particular our distinction between variables for literals and formulae in the above rules, and between `structural' and `logical' rules.

We also have the logical rule $=$ which is obtained by closing the equations below under reflexivity, symmetry, transitivity and by applying context closure. We implicitly assume that it is contained in every system.
\[
\begin{array}{cc}
\begin{array}{cc}
\begin{array}{c}
\mbox{\emph{Commutativity}}\\
\noalign{\bigskip}
\vls[A.B]=[B.A]\\
\vls(A.B)=(B.A)
\end{array}
&\qquad
\begin{array}{c}
\mbox{\emph{Associativity}}\\
\noalign{\bigskip}
\vls[[A.B].C]=[A.[B.C]]\\
\vls((A.B).C)=(A.(B.C)) \\
\end{array} \\
\multispan2{$\begin{array}{c}\noalign{\bigskip}\mbox{\emph{Context closure: }} \mbox{if $A=B$, $\star\in\{\wedge,\vee\}$, then $C\star A = C\star B$}\end{array}$}
\end{array}
&\qquad
\begin{array}{c}
\mbox{\emph{Units}}\\
\noalign{\bigskip}
\vls[A.\fff]=A\\
\vls(A.\ttt)=A\\
\vls[\ttt.\ttt]=\ttt\\
\vls(\fff.\fff)=\fff
\end{array}
\end{array}
\]\smallskip
\end{defi}

\noindent A proof that $=$ is decidable in polynomial time can be found in \cite{BrusGugl:07:On-the-P:fk}. Essentially both formulae are just reduced to some canonical form and then compared. Consequently we often omit occurrences of the $=$ rule in proofs and derivations.

\begin{defi}[Proofs and derivations]\label{PrfDer}
We define derivations and premiss and conclusion functions ($\prem$ and $\conc$ respectively) inductively. 
\begin{enumerate}
\item Each formula $A$ is a derivation with premiss and conclusion $A$.
\item If $\Phi$ and $\Psi$ are derivations and $\star\in\{\wedge,\vee\}$ then $(\Phi\star\Psi)$ is a derivation with premiss $\prem(\Phi)\star\prem(\Psi)$ and conclusion $\conc(\Phi)\star\conc(\Psi)$.  
\item If $\Phi$ and $\Psi$ are derivations and $\vlinf{\rho}{}{\prem(\Psi)}{\conc(\Phi)}$ is an instance of some inference rule $\rho$ then $\vlupsmash{\vlinf{\rho}{}{\Psi}{\Phi}}$ is a derivation with premiss $\prem(\Phi)$ and conclusion $\conc(\Psi)$.
\end{enumerate}

If $\prem(\Phi)\equiv\ttt$ then we call $\Phi$ a \emph{proof}. If $\Phi$ is a derivation where all inference steps are instances of rules in a system $\mathcal{S}$ with premiss $A$, conclusion $B$, we write $\vldownsmash{\vlder{\Phi}{\mathcal{S}}{B}{A}}$. Furthermore, if $A\equiv\ttt$, i.e.\ $\Phi$ is a proof in a system $\mathcal{S}$, we write $\vldownsmash{\vlproof{\Phi}{\mathcal{S}}{B}}$.
\end{defi}

While our structural rules only have atoms in their premisses and conclusions, the notion of derivation above allows us to extend these to arbitrary formulae, as stated in the proposition below. We often use these `generic rules' rather than their full derivations for convenience.
%

\begin{prop}[Generic rules]\label{GenRules}
Each rule below is derivable from $\swi$, $\med$, and its respective atomic structural rule in polynomial time.
\[
\vlinf{\gidown}{}{A\vlor\bar{A}}\ttt \qquad\vlinf{\gwd}{}{A}{\fff}\qquad\vlinf{\gcdown}{}{A}{A\vlor A}\qquad \vlinf{\giu}{}{\fff}{A\vlan\bar{A}} \qquad\vlinf{\gwu}{}{\ttt}{A} \qquad\vlinf{\gcu}{}{A\vlan A}{A}
\]
\end{prop}
\begin{proof}
See \cite{BrunTiu:01:A-Local-:ly} for full proofs. Here we just give an example of the case for contraction, since that is the only structural rule of the sequent calculus that cannot be reduced to atomic form \cite{Brun:03:Two-Rest:mn}. The proof is by induction on the depth of the conclusion of a $\gcdown$ step.
\[
\begin{array}{c}
\vlinf{\gcdown}{}{\vls[A.B]}{\vls[[A.B].[A.B]]} \quad \to\quad \vlinf{=}{}{\vls[\vlinf{\gcdown}{}{A}{A\vlor A} . \vlinf{\gcdown}{}{B}{B\vlor B}]}{\vls[[A.B].[A.B]]}
\\
\vlinf{\gcdown}{}{\vls(A.B)}{\vls[(A.B).(A.B)]} \quad \to \quad \vlinf{\med}{}{\vlinf{\gcdown}{}{A}{A\vlor A} \vlan \vlinf{\gcdown}{}{B}{B\vlor B}}{\vls[(A.B).(A.B)]}
\end{array}
\]
Note that the case for cocontraction is dual to this: one can just flip the derivations upside down and replace every formula with its De Morgan dual.
\end{proof}


\begin{defi}[Complexity]
We define the size $|\Phi|$ of a derivation $\Phi$ to be the number of atom occurrences in $\Phi$. A system $\mathcal{S}$ \emph{polynomially simulates} a system $\mathcal{T}$ if each proof in $\mathcal{T}$ can be polynomially transformed into a proof in $\mathcal{S}$ of the same conclusion.
\end{defi}

\subsection{Atomic flows}
We give only an informal definition of atomic flows here, but refer the reader to \cite{GuglGund:07:Normalis:lr}, \cite{Gund:09:A-Genera:kx} for a formal account of atomic flows. 


\begin{defi}[Atomic flows]
For an $\SKS$ derivation $\Phi$ we define its \emph{atomic flow}, $\flow(\Phi)$, to be the diagram obtained by tracing the path of each atom through the derivation, designating the creation, duplication and destruction of atoms by the following corresponding nodes: 
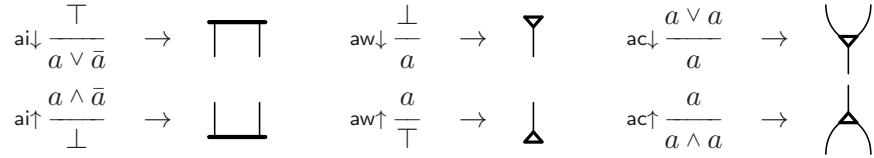
\[
\begin{array}{ccc}
\vlinf{\aid}{}{\vls[a.\bar{a}]}{\ttt}\quad\to\quad \aflower{\afaid{}{}{}{}{}{}} &\qquad \vlinf{\awd}{}{a}\fff \quad\to\quad \aflower{\afawd{}{}{}{}} &\qquad \vlinf{\acd}{}{a}{\vls[a.a]}\quad\to\quad \aflower{\aflower{\afacd{}{}{}{}{}{}}} \\
\vlinf{\aiu}{}{\fff}{\vls(a.\bar{a})}\quad\to\quad \aflower{\aflower{\aflower{\afaiu{}{}{}{}{}{}}}} &\qquad \vlinf{\awu}{}{\ttt}{a} \quad\to\quad \aflower{\aflower{\aflower{\afawu{}{}{}{}}}} &\qquad \vlinf{\acu}{}{\vls(a.a)}{a}\quad\to\quad \aflower{\aflower{\afacu{}{}{}{}{}{}}}
\end{array}\quad
\]
We do not have nodes for $\swi $, $\med $ or $=$ since they do not create, destroy or duplicate any atom occurrences.

More generally an atomic flow, not necessarily of a derivation, is a (vertically) directed graph embedded in the plane generated from the six types of node above. 
 
 Atomic flows are considered equivalent up to continuous deformation preserving the (vertical) ordering of connected edges. Note that edges may be \emph{pending} at either end. 
 
 We define the size of a flow $\phi$, denoted $|\phi|$, to be its number of edges.
\end{defi}

%

In previous works atomic flows have been equipped with a labeling of the edges, or a polarity assignment, for example to avoid the following impossible situation:
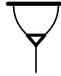
\[
\atomicflow{
(0,0)*{\afacd{}{}{}{}{}{}};
(0,4)*{\afaidnw{}{}};
}
\]
Since we are only concerned with the complexity of flows and their transformations we do not include this extra structure; this does not affect the soundness or termination of our rewriting systems, and in fact is crucial in order to obtain confluence, for which labellings of edges can cause problems.

We do, however, insist that edges are vertically directed and so we are often able to talk about one node being `above' another node. Notice that this order is not generally preserved under deformation, e.g.\ if two nodes are in disconnected components. Whenever we use this notion in arguments it should be clear that it is being used correctly.

\begin{defi}\label{dfn:FlowRewSys}
A \emph{flow rewriting rule} is an ordered pair of flows, written $\phi\to\psi$. A \emph{flow rewriting system} (FRS) is a set of flow rewriting rules. A \emph{one-step reduction} of a flow $\phi$ in an FRS $\mathsf r$ is a flow $\psi$ that is obtainable from $\phi$ by replacing some induced subgraph that is the left hand side of some rule in $\mathsf r$ with its right hand side.
\end{defi}


\begin{defi}\label{GNorm}
We define a graph rewriting system $\GNorm$ on flows in Fig.~\ref{FigRed}.
\end{defi}
\begin{figure}[tbp]
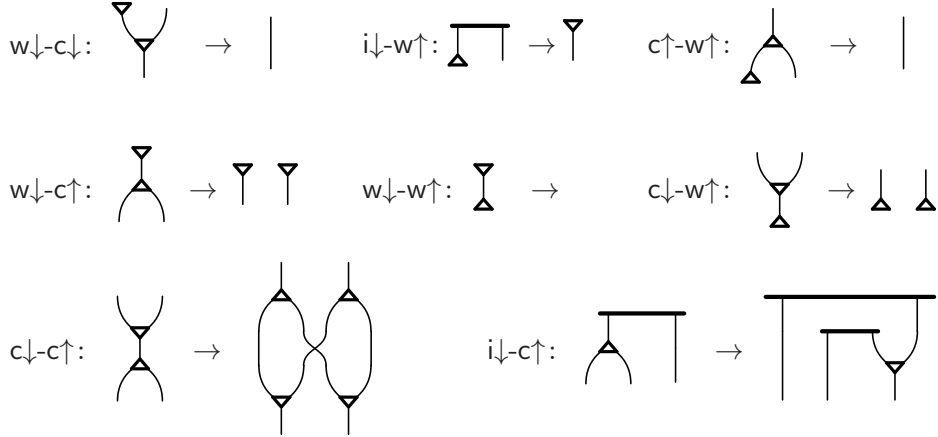

\afnegspace
\[
\begin{array}{cccccc}
\multispan2{$\rwdcd\colon\afraise{\atomicflow{
( 0  ,0)*{\afacd{}{}{}\one{}\two};
(-2  ,4)*{\afawdnw{}{}};
(-3.5,0)*{\invisiblemark};
( 3.5,0)*{\invisiblemark}}}
\to
\atomicflow{
( 0  ,3.3)*{\aflabelright{\one\two}};
( 0  ,3  )*{\afvj6};
(-1.5,0  )*{\invisiblemark};
( 3  ,0  )*{\invisiblemark}}$}&\multispan2{$\qquad\ridwu\colon\atomicflowinv{
( 0  ,0)*{\afaid{}{}{}\one{}{}};
(-2  ,6)*{\afawunw{}{}};
( 3.5,0)*{\invisiblemark}}
\to
\afraise{\atomicflowinv{
(0  ,0)*{\afawd{}{}{}\one};
(1.5,0)*{\invisiblemark}}}$}&\multispan2{$\qquad\rcuwu\colon\aflower{\atomicflowinv{
( 0  ,0)*{\afacu{}{}{}\two{}\one};
(-2  ,6)*{\afawunw{}{}};
( 3.5,0)*{\invisiblemark}}}
\to\quad
\atomicflow{
(0,3  )*{\afvj6};
(0,3.3)*{\aflabelright{\one\two}};
(3,0  )*{\invisiblemark}}$}
\\
\multispan2{$\rwdcu\colon\afraise{\atomicflow{
( 0  ,-4)*{\afacu\one{}{}\two{}{}};
( 0  , 0)*{\afawdnw{}{}};
(-3.5, 0)*{\invisiblemark};
( 3.5, 0)*{\invisiblemark}}}
\to
\afraise{\atomicflow{
(-2  ,-4)*{\afawd{}{}\one{}};
( 2  ,-4)*{\afawd{}{}{}\two};
(-3.5, 0)*{\invisiblemark};
( 3.5, 0)*{\invisiblemark}}}$}&\multispan2{$\qquad\rwdwu\colon\atomicflow{
( 0  ,6)*{\afawd{}{}{}{}};
( 0  ,0)*{\afawunw{}{}{}{}};
(-1.5,0)*{\invisiblemark};
( 1.5,0)*{\invisiblemark}}
\to
\atomicflow{}$}&\multispan2{$\qquad\rcdwu\colon\aflower{\atomicflowinv{
( 0  ,-4)*{\afacd\one{}{}\two{}{}};
( 0  , 2)*{\afawunw{}{}};
(-3.5, 0)*{\invisiblemark};
( 3.5, 0)*{\invisiblemark}}}
\to
\aflower{\atomicflowinv{
(-2  ,-4)*{\afawu{}{}\one{}};
( 2  ,-4)*{\afawu{}{}{}\two};
(-3.5, 0)*{\invisiblemark};
( 3.5, 0)*{\invisiblemark}}}$}
\\
\multispan3{$\rcdcu\colon\atomicflow{
( 0,6)*{\afacd\one{}{}\two{}{}};
( 0,0)*{\afacunw\three{}{}\four};
(-4,0)*{\invisiblemark};
( 4,0)*{\invisiblemark}}
\to\quad
\atomicflow{
(0,12)*{\afacu{}{}{}{}\one{}};
(6,12)*{\afacu{}{}{}{}{}\two};
( 0,0)*{\afacd{}{}{}{}\three{}};
( 6,0)*{\afacd{}{}{}{}{}\four};
(-2,6)*{\afvj4};
( 8,6)*{\afvj4};
( 3,6)*{\afex24}}\qquad$}&\multispan3{$\qquad\ridcu\colon\afraise{\atomicflowinv{
(   6,6)*{\afvju4{}\three{}{}};
(   3,0)*{\afaidex {}{}{}{}{}{}32};
(   0,6)*{\afacunw\one{}{}\two};
(-3.5,0)*{\invisiblemark};
( 7.5,0)*{\invisiblemark}}}
\to
\afraise{\atomicflowinv{
(  10,8)*{\afacd{}{}{}{}{}\three};
(   0,8)*{\afvju8\one{}};
(   4,8)*{\afvju8{}\two};
(   6,4)*{\afaidnw{}{}};
(   6,0)*{\afaidex{}{}{}{}{}{}31}}}$}
\end{array}
\]
\afnegspace
\caption{Local rewriting rules for the system $\GNorm$.}
\label{FigRed}
\end{figure}

The system $\GNorm$ is essentially the system $\mathsf{c}\cup\mathsf{w}$ in \cite{GuglGund:07:Normalis:lr}, without the rules for $\giu$. The proof of termination that follows is similar to that for cycle-free flows in \cite{GuglGund:07:Normalis:lr}, and the proof of confluence is similar to that in \cite{GuglGundStra::Breaking:uq}.

\newcommand{\normform}{\hspace{-2pt}\downarrow}

\begin{nota}
Let $\mathsf{r}$ be a flow rewriting system. We use the following notation:
\begin{itemize}
\item We write $\phi\underset{\mathsf{r}}{\rightarrow}\psi$ if there is a one-step reduction from a flow $\phi$ to a flow $\psi$ using a rule in $\mathsf{r}$.
\item  We denote by $\underset{\mathsf{r}}{\overset{*}{\rightarrow}}$ the reflexive transitive closure of $\underset{\mathsf{r}}{\rightarrow}$. 
\item  If a flow $\phi$ has a unique normal form under $\mathsf r$ we denote it by $\phi\normform_\mathsf{r}$.\footnote{Note that we are using downward arrows in both the names of deep inference rules and to denote normal forms under rewriting systems. Unfortunately both notations are standard in their respective literature, however there should be no ambiguity in their usage so hopefully this will cause little confusion.}
\end{itemize}
In all cases we might omit the subscript $\mathsf{r}$ if it is clear from context.
\end{nota}

\begin{exa}
We give an example of a flow associated with a derivation in Fig.~\ref{ExRed}, as well as a reduction under $\GNorm$, as defined in Dfn.~\ref{GNorm}, applying $\gwu$ rules first.

The first equality follows by the definition of a flow, the second by deformation and the final by definition again. The intermediate steps are as follows:
\begin{enumerate}
\item Apply $\rcdwu$ on the left and $\rcdcu$ on the right.
\item Apply $\rwdwu$ on the left, $\ridwu$ in the middle and $\ridcu$ on the right.
\item Apply $\rwdcu$ on the left.
\item Apply $\rwdcd$ twice on the left.
\end{enumerate}

The use of colours in the initial and final flow identifies which edges corresponds to which atoms; in the intermediate flows the colours should aid the reader in reconstructing the corresponding transformations on the derivation. 
\end{exa}

\begin{figure}
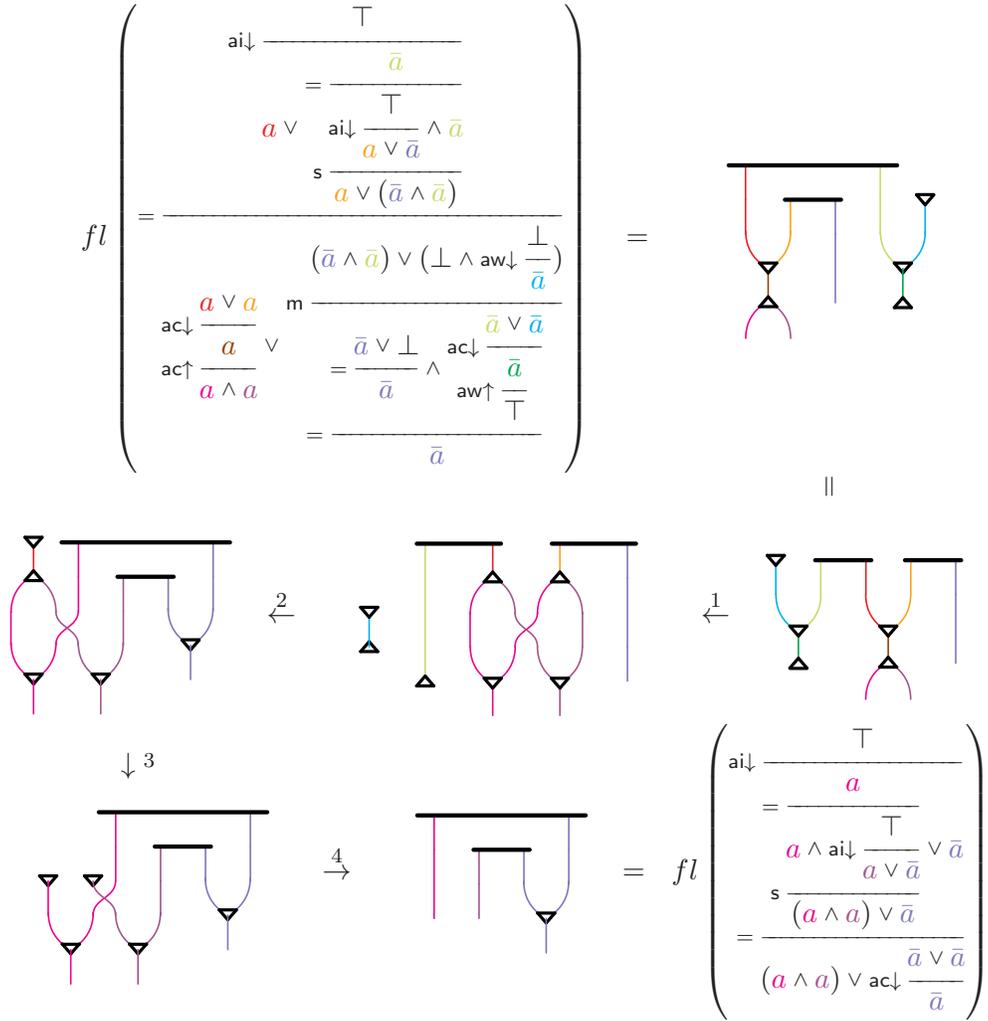

\afnegspace
\[
\begin{array}{ccc}
\noalign{\bigskip}
\multispan2{\hfill$\flow\left(\vlderivation{
\vlin{=}{}{\vls[\vlderivation{\vlin{\acu}{}{\vls(\MG{a}.\DO{a})}{\vlin{\acd}{}{\RS{a}}{\vlhy{\vls[\RD{a}.\YO{a}]}}}} . \vlderivation{
\vlin{=}{}{\PW{\bar{a}}}{
\vlin{\med}{}{\vls(\vlinf{=}{}{\PW{\bar{a}}}{\PW{\bar{a}}\vlor\fff} . \vlderivation{\vlin{\awu}{}{\ttt}{\vlin{\acd}{}{\GR{\bar{a}}}{\vlhy{\vls[\SG{\bar{a}}.\PB{\bar{a}}]}}}})}{\vlhy{\vls[(\PW{\bar{a}}.\SG{\bar{a}}).(\fff . \vlinf{\awd}{}{\PB{\bar{a}}}{\fff} ) ]}}
}
}]}{
\vlin{\aid}{}{\vls[\RD{a} . \vlderivation{
\vlin{\swi}{}{\vls[\YO{a} . (\PW{\bar{a}} . \SG{\bar{a}} ) ]}{
\vlin{=}{}{\vls(\vlinf{\aid}{}{\YO{a}\vlor\PW{\bar{a}}}{\ttt} . \SG{\bar{a}})}{\vlhy{\SG{\bar{a}}}}
}
} ]}{\vlhy\ttt}
}
}\right) \quad = $}
&\atomicflow{
(12,16)*{\afaidexcol{}{}{}{}{}{}62{Red}{SpringGreen}};
(12,12)*{\afaidcol{}{}{}{}{}{}{YellowOrange}{Periwinkle}};
(18,10)*{\afvjcol4{SpringGreen}{}{}{}{}{}};
(8,2)*{\afvjcol4{RawSienna}{}{}{}{}{}};
(6,10)*{\afvjcol4{Red}{}{}{}{}{}};
(8,4)*{\afacdnwcol{}{}{}{}{Red}{YellowOrange}};
(14,4)*{\afvjcol8{Periwinkle}{}{}{}{}{}};
(22,12)*{\afawdcol{}{}{}{}{}{ProcessBlue}};
(20,4)*{\afacdnwcol{}{}{}{}{SpringGreen}{ProcessBlue}};
(20,0)*{\afawucol{}{}{}{}{}{Green}};
(8,-2)*{\afacunwcol{}{}{}{}{Magenta}{DarkOrchid}};
}
\\
	&	&	\begin{sideways} = \end{sideways} \\
\atomicflow{
(2,16)*{\afawdcol{}{}{}{}{}{Red}};
(2,10)*{\afacunwcol{}{}{}{}{Magenta}{DarkOrchid}};
(12,16)*{\afaidexcol{}{}{}{}{}{}62{Magenta}{Periwinkle}};
(6,10)*{\afvjcol4{Magenta}{}{}{}{}{}};
(5,6)*{\afexcol24{DarkOrchid}{Magenta}};
(12,12)*{\afaidcol{}{}{}{}{}{}{DarkOrchid}{Periwinkle}};
(18,10)*{\afvjcol4{Periwinkle}{}{}{}{}{}};
(16,4)*{\afacdnwcol{}{}{}{}{Periwinkle}{Periwinkle}};
(16,2)*{\afvjcol4{Periwinkle}{}{}{}{}{}};
(10,6)*{\afvjcol4{DarkOrchid}{}{}{}{}{}{}{}{}{}};
(0,6)*{\afvjcol4{Magenta}{}{}{}{}{}};
(2,0)*{\afacdnwcol{}{}{}{}{Magenta}{Magenta}};
(8,0)*{\afacdnwcol{}{}{}{}{DarkOrchid}{DarkOrchid}};
(2,-2)*{\afvjcol4{Magenta}{}{}{}{}{}};
(8,-2)*{\afvjcol4{DarkOrchid}{}{}{}{}{}};
}
\quad\overset{2}{\leftarrow}
&\quad\atomicflow{
(8,16)*{\afaidexcol{}{}{}{}{}{}32{SpringGreen}{Red}};
(20,16)*{\afaidexcol{}{}{}{}{}{}32{YellowOrange}{Periwinkle}};
(23,8)*{\afvjcol8{Periwinkle}{}{}{}{}{}};
(23,2)*{\afvjcol4{Periwinkle}{}{}{}{}{}};
(5,8)*{\afvjcol8{SpringGreen}{}{}{}{}{}};
(5,0)*{\afawucol{}{}{}{}{}{SpringGreen}};
(17,10)*{\afacunwcol{}{}{}{}{Magenta}{DarkOrchid}};
(11,10)*{\afacunwcol{}{}{}{}{Magenta}{DarkOrchid}};
(14,6)*{\afexcol24{DarkOrchid}{Magenta}};
(19,6)*{\afvjcol4{DarkOrchid}{}{}{}{}{}};
(9,6)*{\afvjcol4{Magenta}{}{}{}{}{}};
(11,0)*{\afacdcol{}{}{}{}{}{}{Magenta}{Magenta}{Magenta}};
(17,0)*{\afacdcol{}{}{}{}{}{}{DarkOrchid}{DarkOrchid}{DarkOrchid}};
(0,2)*{\afawunw{}{}{}{}{}{}};
(0,8)*{\afawdcol{}{}{}{}{}{ProcessBlue}};
}
&\overset{1}{\leftarrow}\quad\atomicflow{
(8,16)*{\afaidcol{}{}{}{}{}{}{SpringGreen}{Red}};
(4,8)*{\afacdcol{}{}{}{}{}{}{ProcessBlue}{SpringGreen}{Green}};
(16,16)*{\afaidcol{}{}{}{}{}{}{YellowOrange}{Periwinkle}};
(12,8)*{\afacdcol{}{}{}{}{}{}{Red}{YellowOrange}{RawSienna}};
(18,8)*{\afvjcol8{Periwinkle}{}{}{}{}{}};
(2,16)*{\afawdcol{}{}{}{}{}{ProcessBlue}};
(12,2)*{\afacunwcol{}{}{}{}{Magenta}{DarkOrchid}};
(4,2)*{\afawunw{}{}{}{}};
}
\\


\atomicflow{
(8,22)*{\downarrow\overset{3}{} };
(0,8)*{\afawdnw{}{}{}{}{}{}};
(4,8)*{\afawdnw{}{}{}{}{}{}};
(12,16)*{\afaidexcol{}{}{}{}{}{}62{Magenta}{Periwinkle}};
(6,10)*{\afvjcol4{Magenta}{}{}{}{}{}};
(5,6)*{\afexcol24{DarkOrchid}{Magenta}};
(12,12)*{\afaidcol{}{}{}{}{}{}{DarkOrchid}{Periwinkle}};
(18,10)*{\afvjcol4{Periwinkle}{}{}{}{}{}};
(16,4)*{\afacdnwcol{}{}{}{}{Periwinkle}{Periwinkle}};
(16,2)*{\afvjcol4{Periwinkle}{}{}{}{}{}};
(10,6)*{\afvjcol4{DarkOrchid}{}{}{}{}{}};
(0,6)*{\afvjcol4{Magenta}{}{}{}{}{}};
(2,0)*{\afacdnwcol{}{}{}{}{Magenta}{Magenta}};
(8,0)*{\afacdnwcol{}{}{}{}{DarkOrchid}{DarkOrchid}};
(2,-2)*{\afvjcol4{Magenta}{}{}{}{}{}};
(8,-2)*{\afvjcol4{DarkOrchid}{}{}{}{}{}};
}
&\overset{4}{\to}\qquad\atomicflow{
(12,16)*{\afaidexcol{}{}{}{}{}{}62{Magenta}{Periwinkle}};
(12,12)*{\afaidcol{}{}{}{}{}{}{DarkOrchid}{Periwinkle}};
(18,10)*{\afvjcol4{Periwinkle}{}{}{}{}{}};
(16,4)*{\afacdnwcol{}{}{}{}{Periwinkle}{Periwinkle}};
(16,2)*{\afvjcol4{Periwinkle}{}{}{}{}{}};
(10,6)*{\afvjcol4{DarkOrchid}{}{}{}{}{}};
(6,8)*{\afvjcol8{Magenta}{}{}{}{}{}};
}
\quad =
&\flow\left(\vlderivation{
\vlin{=}{}{\vls[(\MG{a}.\DO{a}).\vlinf{\acd}{}{\PW{\bar{a}}}{\vls[\PW{\bar{a}}.\PW{\bar{a}}]}]}{
\vlin{\aid}{}{\vls[\vlderivation{
\vlin{\swi}{}{\vls[(\MG{a}.\DO{a}).\PW{\bar{a}}]}{
\vlin{=}{}{\vls(\MG{a}.\vlinf{\aid}{}{\DO{a}\vlor\PW{\bar{a}}}\ttt)}{\vlhy{\MG{a}}}
}
}.\PW{\bar{a}}]}{\vlhy{\ttt}}
}
}\right)
\end{array}
\]
\afnegspace
\caption{A proof, its flow and a reduction under $\GNorm$.}
\label{ExRed}
\end{figure}

We now proceed to prove that reducing under $\GNorm$ is terminating and confluent. For this we use the usual critical pair lemma, that one only needs to check that the one-step contracta of every overlapping pair of redexes are joinable in order to conclude local confluence. It is simple to see that this is still valid in the flow rewriting setting, since the only other case to check for flows is when the redexes are disjoint, which is trivial. We point out that Newman's lemma, that any locally confluent terminating system is confluent, is true more generally for any binary relation on any set, and so indeed holds in this setting.

%
\begin{thm}\label{NormTermConf}
$\underset{\GNorm}{\longrightarrow}$ is terminating and confluent. 
\end{thm}
\begin{proof}
For a node $\nu$ in a flow $\phi$, let $d(\nu,\phi)$ be the distance of $\nu$ from the top of $\phi$, i.e.\ the minimum length of a (vertically) directed path from $\nu$ to an $\aid $ node, an $\awd $ node or an edge with upper end pending. For an atomic structural rule $\rho$, let $D(\rho,\phi)$ be the sequence of natural numbers that counts how many $\rho$ nodes in $\phi$ have each $d$-value, i.e.\ the sequence $(n_i)$ such that, for each $i$, $n_i$ is the number of $\rho$ nodes $\nu$ in $\phi$ with $d(\nu,\phi) = i$, and consider the lexicographical ordering $<$ on such sequences, with $(m_i) < (n_i)$ just if $m_k < n_k$ for some $k$ and $m_i \leq n_i$ for $i>k$.

Clearly the rules $\rcdcu$, $\ridcu$ and $\rwdcu$ strictly reduce the $D(\acu ,\cdot)$ value of a flow, while the other rules of $\GNorm$ (as well as $\rwdcu$) strictly reduce a flow's size, while not increasing the $D(\acu ,\cdot)$ value. Therefore each application of a $\GNorm$ rule strictly reduces the lexicographical product $D(\acu ,\cdot) \times |\cdot|$.

Since $\GNorm$ is terminating, it suffices to check local confluence for the critical pairs, which are the following by inspection:
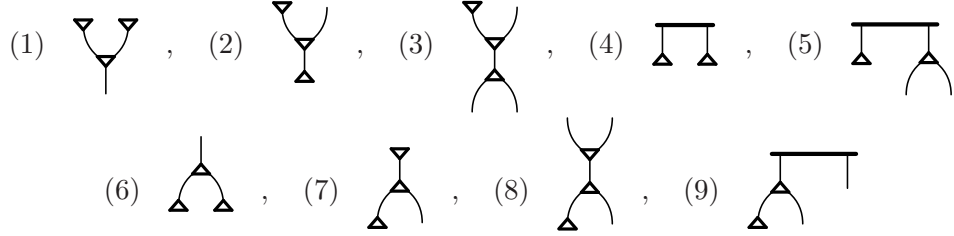
\[
\begin{array}{c}
(1)\quad \atomicflow{
(0,0)*{\afacd{}{}{}{}{}{}};
(-2,4)*{\afawdnw{}{}};
(2,4)*{\afawdnw{}{}};
}\quad,\quad
(2)\quad \atomicflow{
(0,0)*{\afacd{}{}{}{}{}{}};
(-2,4)*{\afawdnw{}{}};
(0,-6)*{\afawunw{}{}};
}\quad,\quad
(3)\quad\atomicflow{
(0,0)*{\afacd{}{}{}{}{}{}};
(-2,4)*{\afawdnw{}{}};
(0,-6)*{\afacunw{}{}{}{}};	
}\quad,\quad
(4) \quad  \atomicflow{
(0,0)*{\afaid{}{}{}{}{}{}};
(-2,-6)*{\afawunw{}{}};
(2,-6)*{\afawunw{}{}};
}\quad,\quad
(5)\quad \atomicflow{
(0,0)*{\afaidex{}{}{}{}{}{}32};
(-3,-6)*{\afawunw{}{}};
(3,-6)*{\afacunw{}{}{}{}};
}
\\
(6)\quad \atomicflow{
(0,0)*{\afacu{}{}{}{}{}{}};
(-2,-6)*{\afawunw{}{}};
(2,-6)*{\afawunw{}{}};
}\quad,\quad
(7)\quad \atomicflow{
(0,0)*{\afacu{}{}{}{}{}{}};
(-2,-6)*{\afawunw{}{}};
(0,4)*{\afawdnw{}{}};
}\quad,\quad
(8)\quad \atomicflow{
(0,0)*{\afacu{}{}{}{}{}{}};
(-2,-6)*{\afawunw{}{}};
(0,4)*{\afacdnw{}{}{}{}};
}\quad,\quad
(9)\quad \atomicflow{
(0,0)*{\afaidex{}{}{}{}{}{}{3}{2}};
(-3,-6)*{\afacunw{}{}{}{}};
(-5,-10)*{\afawunw{}{}};
}
\end{array}
\]
Note that every other overlapping pair can be deformed so that each rule application trivially commutes. We consider each case below.
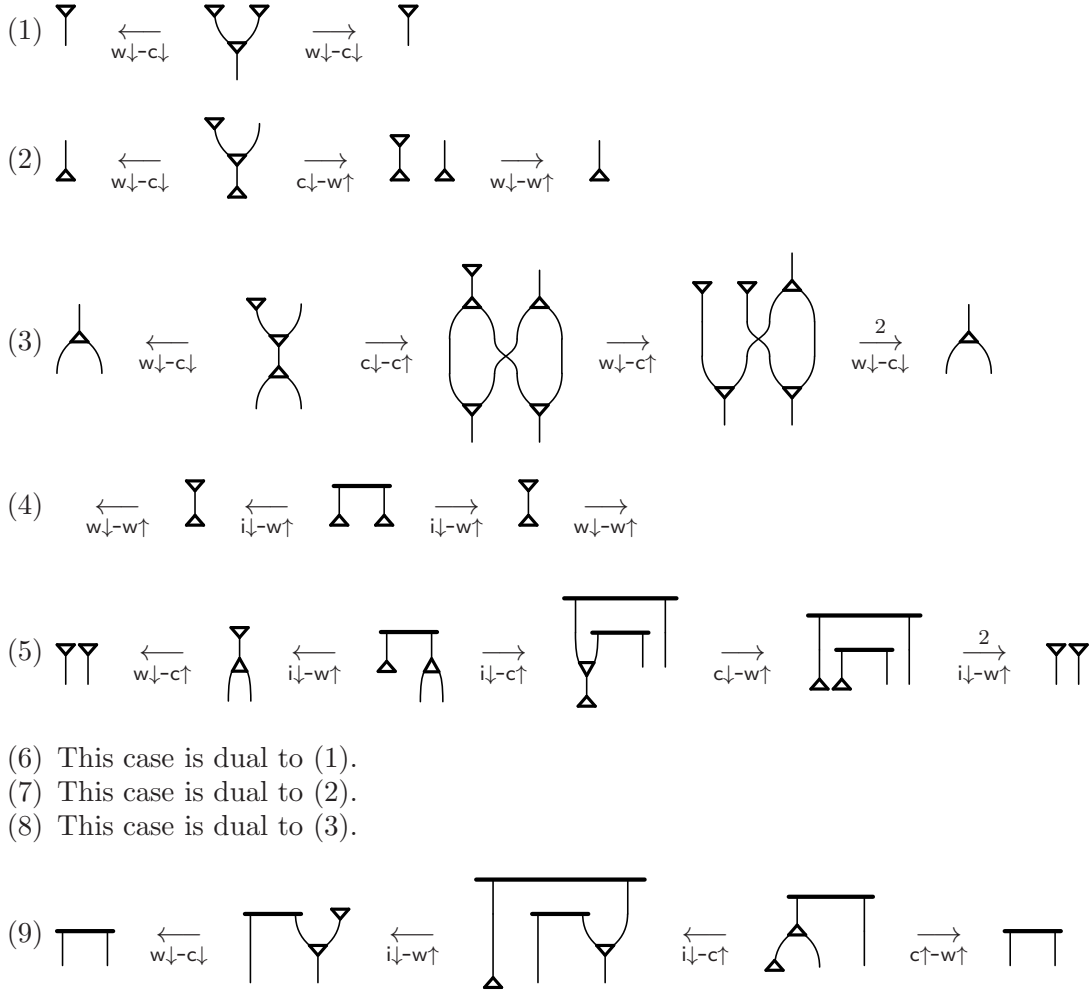
\begin{enumerate}
\item $ \afraise{\atomicflow{
(0,0)*{\afawd{}{}{}{}}
}} \quad\underset{\rwdcd}{\longleftarrow}\quad \atomicflow{
(0,0)*{\afacd{}{}{}{}{}{}};
(-2,4)*{\afawdnw{}{}};
(2,4)*{\afawdnw{}{}};
} \quad \underset{\rwdcd}{\longrightarrow} \quad \afraise{\atomicflow{
(0,0)*{\afawd{}{}{}{}}
}}$
\item $ \aflower{\atomicflow{
(0,0)*{\afawu{}{}{}{}}
}} \quad \underset{\rwdcd}{\longleftarrow} \quad \atomicflow{
(0,0)*{\afacd{}{}{}{}{}{}};
(-2,4)*{\afawdnw{}{}};
(0,-6)*{\afawunw{}{}};
}\quad \underset{\rcdwu}{\longrightarrow} \quad \atomicflow{
(0,0)*{\afawu{}{}{}{}};
(0,4)*{\afawdnw{}{}};
(4,0)*{\afawu{}{}{}{}};
} \quad \underset{\rwdwu}{\longrightarrow} \quad \aflower{\atomicflow{(0,0)*{\afawu{}{}{}{}}}}$
\item $ \atomicflow{
(0,0)*{\afacu{}{}{}{}{}{}}
} 
\quad \underset{\rwdcd}{\longleftarrow}\quad 
\atomicflow{
( 0,6)*{\afacd\one{}{}\two{}{}};
( 0,0)*{\afacunw\three{}{}\four};
(-4,0)*{\invisiblemark};
( 4,0)*{\invisiblemark};
(-2,10)*{\afawdnw{}{}};
} 
\quad \underset{\rcdcu}{\longrightarrow} \quad 
\atomicflow{
(0,12)*{\afacu{}{}{}{}\one{}};
(6,12)*{\afacu{}{}{}{}{}\two};
( 0,0)*{\afacd{}{}{}{}\three{}};
( 6,0)*{\afacd{}{}{}{}{}\four};
(-2,6)*{\afvj4};
( 8,6)*{\afvj4};
( 3,6)*{\afex24};
(0,16)*{\afawdnw{}{}};
}
\quad \underset{\rwdcu}{\longrightarrow} \quad
\atomicflow{
(6,12)*{\afacu{}{}{}{}{}\two};
(-2,12)*{\afawd{}{}{}{}};
(2,12)*{\afawd{}{}{}{}};
( 0,0)*{\afacd{}{}{}{}\three{}};
( 6,0)*{\afacd{}{}{}{}{}\four};
(-2,6)*{\afvj4};
( 8,6)*{\afvj4};
( 3,6)*{\afex24};
} 
\quad \overset{2}{\underset{\rwdcd}{\longrightarrow}}\quad 
\atomicflow{
(0,0)*{\afacu{}{}{}{}{}{}}
} $
\item $ \quad \underset{\rwdwu}{\longleftarrow}\quad
\atomicflow{
(0,0)*{\afawu{}{}{}{}};
(0,4)*{\afawdnw{}{}};
}
\quad \underset{\ridwu}{\longleftarrow} \quad
 \atomicflow{
(0,0)*{\afaid{}{}{}{}{}{}};
(-2,-6)*{\afawunw{}{}};
(2,-6)*{\afawunw{}{}};
}
\quad \underset{\ridwu}{\longrightarrow} \quad
\atomicflow{
(0,0)*{\afawu{}{}{}{}};
(0,4)*{\afawdnw{}{}};
}
\quad \underset{\rwdwu}{\longrightarrow} \quad	 $
\item $\atomicflow{
(0,0)*{\afawd{}{}{}{}};
(2,0)*{\afawd{}{}{}{}};
}
\quad \underset{\rwdcu}{\longleftarrow}\quad
\atomicflow{
(0,0)*{\afacuex{}{}{}{}{}{}24};
(0,4)*{\afawdnw{}{}};
}
\quad \underset{\ridwu}{\longleftarrow} \quad
\atomicflow{
(0,0)*{\afaid{}{}{}{}{}{}};
(-2,-6)*{\afawunw{}{}};
(2,-6)*{\afacunwex{}{}{}{}24};
}
\quad\underset{\ridcu}{\longrightarrow} \quad
\atomicflow{
(0,0)*{\afaidex{}{}{}{}{}{}{4}{2}};
(0,-4)*{\afaidnw{}{}{}{}{}{}};
(-3,-8)*{\afacdex{}{}{}{}{}{}24};
(-3,-14)*{\afawunw{}{}};
(2,-6)*{\afvj{4}};
(4,-6)*{\afvj{4}};
}
\quad \underset{\rcdwu}{\longrightarrow} \quad
\atomicflow{
(0,0)*{\afaidex{}{}{}{}{}{}{4}{2}};
(0,-4)*{\afaidnw{}{}{}{}{}{}};
(2,-6)*{\afvj{4}};
(4,-6)*{\afvj{4}};
(-2,-8)*{\afawu{}{}{}{}};
(-4,-8)*{\afawu{}{}{}{}};
} 
\quad \overset{2}{\underset{\ridwu}{\longrightarrow}} \quad
\atomicflow{
(0,0)*{\afawd{}{}{}{}};
(2,0)*{\afawd{}{}{}{}};
} $
\item This case is dual to (1).
\item This case is dual to (2).
\item This case is dual to (3).
\item $\atomicflow{
(0,0)*{\afaid{}{}{}{}{}{}}
} 
\quad \underset{\rwdcd}{\longleftarrow} \quad
\atomicflowinv{
(   6,4)*{\afaidnw{}{}};
(   4,8)*{\afvju8\one{}};
(  10,8)*{\afacd{}{}{}{}{}\three};
(12,4)*{\afawdnw{}{}}
}
\quad \underset{\ridwu}{\longleftarrow}\quad 
\atomicflowinv{
(  10,8)*{\afacd{}{}{}{}{}\three};
(   0,8)*{\afvju8\one{}};
(   4,8)*{\afvju8{}\two};
(   6,4)*{\afaidnw{}{}};
(   6,0)*{\afaidex{}{}{}{}{}{}31};
(0,14)*{\afawunw{}{}}
}
\quad \underset{\ridcu}{\longleftarrow} \quad
\atomicflow{
(0,0)*{\afaidex{}{}{}{}{}{}{3}{2}};
(-3,-6)*{\afacunw{}{}{}{}};
(-5,-10)*{\afawunw{}{}};
(3,-6)*{\afvj4}
}
\quad \underset{\rcuwu}{\longrightarrow} \quad
\atomicflow{
(0,0)*{\afaid{}{}{}{}{}{}};
} $
\end{enumerate}
The cases where we appeal to duality follow by simply flipping the indicated reductions upside down and relabeling nodes and reduction steps appropriately.
\end{proof}

The significance of the rewriting system $\GNorm$ is evident from the following results, where we show that reduction steps correspond to sound manipulations of $\SKS$ derivations.

\begin{defi}\label{RelLiftPoly}
If $\mathcal{R}$ is a relation on flows we say that $\mathcal{R}$ \emph{lifts polynomially} to $\SKS$ if, whenever $(\flow(\Phi) , \psi) \in \mathcal{R}$, we can construct a derivation $\Psi$ in time polynomial in $|\Phi|+|\psi|$ with the same premiss and conclusion as $\Phi$ and flow $\psi$. If $f$ is a function on flows then we say that $f$ lifts polynomially to $\SKS$ just if the relation $f(\cdot)=\cdot$ lifts polynomially to $\SKS$.
\end{defi}

Sometimes we simply say that an individual flow rewrite rule is \emph{sound} rather than saying that it lifts polynomially to $\SKS$.

\begin{rem}
In fact a derivation can always be manipulated (preserving premiss, conclusion and flow) so that it has size at most polynomial in the size of its flow, as (essentially) shown in \cite{DBLP:conf/rta/Das13}, so the dependence on size of derivation in Dfn.\ \ref{RelLiftPoly} is somewhat redundant. However this is beyond the scope of this work, and it does no harm for us to include this dependence.
\end{rem}

\begin{thm}\label{PolyLift}
$\underset{\GNorm}{\longrightarrow}$ lifts polynomially to $\SKS$.
\end{thm}
\begin{proof}
See \cite{GuglGund:07:Normalis:lr}. Essentially the proof shows that each local rewrite step on a flow of a derivation induces a sound manipulation of polynomial size on that derivation. We give as an example the case for $\rcdcu$, since that rule is of particular interest in this work.

Let $\xi\vlhole$ denote an arbitrary context, i.e.\ a formula with a single hole occurring in place of a subformula, and let $\xi\{A\}$ denote the result of substituting a formula $A$ for the hole in $\xi\vlhole$. The manipulation is as follows:
\[
\rcdcu\colon \quad 
\vlderivation{
\vldd{\Phi}{}{\zeta\left\{ \vlinf{\acu}{}{a\vlan a}{a}  \right\}}{\vlhy{\xi\left\{  \vlinf{\acd}{}{a}{a\vlor a} \right\} }}
}
\qquad \to \qquad 
 \vlderivation{
\vldd{\Phi\{( a/(a\vlan a) \}}{}{ \zeta\left\{   a\vlan a  \right\} }{\vlhy{ \xi\left\{  \vlinf{\med}{}{\vls(\vlinf{\acd}{}{a}{\vls[a.a]} . \vlinf{\acd}{}{a}{\vls[a.a]} )}{\vlinf{\acu}{}{a\vlan a}{a} \vlor \vlinf{\acu}{}{a\vlan a}{a}}   \right\} }}
}
\]
where $\Phi\{a/(a\vlan a)  \}$ is obtained by replacing $a$ by $(a\vlan a)$ everywhere in $\Phi$. 
\end{proof}


\begin{cor}\label{FlowTransProofs}
Given an $\SKS$ derivation $\Phi$ and a reduction $\flow(\Phi)= \phi_1 \underset{\GNorm}{\longrightarrow} \vldots \underset{\GNorm}{\longrightarrow} \phi_n = \flow(\Phi)\normform_\GNorm$ we can construct an $\SKS$ derivation with the same premiss and conclusion as $\Phi$ and with flow $\flow(\Phi)\normform_\GNorm$ in time polynomial in $|\Phi| + \sum\limits_{i=1}^{n} |\phi_i |$.
\end{cor}
\begin{proof}
By induction on the length $n$ of a $\GNorm$ derivation. 
\end{proof}



\subsection{Reduction strategies}
We analyse the complexity of normalising a flow under $\GNorm$, presenting a class of reduction strategies that optimise the size of a $\underset{\GNorm}{\longrightarrow}$ derivation from a flow to its normal form, up to a polynomial. Our aim is to prove the following theorem:
\begin{thm}\label{NormFormLiftsPoly}
The function $\cdot\normform_\GNorm$ mapping a flow to its unique normal form under $\underset{\GNorm}{\longrightarrow}$ lifts polynomially to $\SKS$.
\end{thm}
The result follows from Cor.~\ref{FlowTransProofs} if we can find appropriate reductions with size only polynomially dependent on the initial derivation and normal form of its flow.


\begin{exa}\label{ExRedStrat}
Consider the following flow,
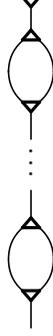
\[
\atomicflow{
(0,30)*{\afawd{}{}{}{}};
(0,24)*{\afacunw{}{}{}{}};
(0,18)*{\afacd{}{}{}{}{}{}};
(0,12)*{\vlvdots};
(0,4)*{\afacu{}{}{}{}{}{}};
(0,-4)*{\afacd{}{}{}{}{}{}};
}
\]
where there are $n$ $\acu $ nodes. Notice that this flow has normal form $\afraise{\atomicflow{
(0,0)*{\afawd{}{}{}{}};
}}$, but the complexity of a derivation witnessing this can vary significantly. If we apply $\rcdcu$ steps first then the length of a derivation to normal form is $\Theta(2^n)$, whereas applying $\{\rwdcd, \rwdcu \}$ steps first results in length $\Theta(n)$. These bounds follow from later results in this section, but are not difficult to prove directly.
\end{exa}

In fact, this sort of unnecessary exponential blowup can always be avoided by applying `weakening' rules first.


\newcommand{\weak}{\mathsf{wk}}
\newcommand{\cont}{\mathsf{cont}}
\begin{defi}
We define the following subsystems of $\GNorm$:
\begin{enumerate}
\item $\weak = \{\rwdcd, \ridwu, \rcuwu, \rwdcu,\rwdwu, \rcdwu\}$.
\item $\cont = \GNorm\setminus\weak = \{\rcdcu,\ridcu\}$.
\end{enumerate}
\end{defi}

\begin{prop}\label{SizeBoundsWkContRed}
Given a flow $\phi$ we have: 
\begin{enumerate}
\item If $\phi\overset{}{\underset{\weak}{\rightarrow}}\psi$ then $|\phi|>|\psi|$.
\item If $\phi\overset{}{\underset{\cont}{\rightarrow}}\psi$ then $|\phi|<|\psi|$.
\end{enumerate}
\end{prop}

\begin{lem}\label{WkContStrat}
Given a flow $\phi$ we have $(\phi\normform_\weak)\normform_\cont = \phi\normform_\GNorm$.
\end{lem}
\begin{proof}
By inspecting the rules of $\cont$, we observe that if there are no $\weak$-redexes in a flow $\psi$ and $\psi\underset{\cont}{\rightarrow}\theta$ then there are no $\weak$-redexes in $\theta$. By induction on the length of a $\underset{\cont}{\rightarrow}$-derivation we then have that there are no $\weak$-redexes in $(\phi\normform_\weak)\normform_\cont$. But since there are also no $\cont$-redexes, by definition of normal form, and $\weak\cup\cont=\GNorm$ we can conclude that $(\phi\normform_\weak)\normform_\cont$ is already in normal form for $\GNorm$.
\end{proof}

We can now give a proof of the main theorem of this section.

\begin{proof}[Proof of Thm.~\ref{NormFormLiftsPoly}]
Let $\Phi$ be an $\SKS$ derivation with flow $\phi$ whose normal form under $\GNorm$ is $\psi$, and fix a derivation $\phi= \phi_1 \underset{\weak}{\rightarrow} \vldots \underset{\weak}{\rightarrow} \phi_m = \psi_1 \underset{\cont}{\rightarrow} \vldots \underset{\cont}{\rightarrow} \psi_n = \psi$, which must exist by Lemma~\ref{WkContStrat}.

By Cor.~\ref{FlowTransProofs} we can construct an $\SKS$ derivation $\Psi$ with same premiss and conclusion as $\Phi$ and flow $\psi$ in time polynomial in $|\Phi| + \sum\limits_{i=1}^{m}|\phi_i| + \sum\limits_{j=1}^n |\psi_i|$. However, by Prop.~\ref{SizeBoundsWkContRed} we have that $|\phi_i|< |\phi_1|=|\phi|\leq|\Phi|$ for all $i$ and $|\psi_j|< |\psi_n| = |\psi|$ for all $j$, and so this construction can be done in time polynomial in $(m+n)(|\Phi| + |\psi| )$. 

Finally, notice also by Prop.~\ref{SizeBoundsWkContRed} that each $\underset{\weak}{\rightarrow}$ step decreases the size of a flow, so $m$ is bounded above by $|\phi|$ (and so also $|\Phi|$), and each $\cont$ step increases the size of a flow, so $n$ is bounded above by $|\psi|$, whence the result follows.
\end{proof}


\section{Complexity of normal forms}
In this section we specialise previous results to $\aSKS$ flows, reducing the complexity of $\GNorm$ reduction to counting the number of certain paths in the initial flow. We then give two simple ways of estimating this number, which we use in later sections to obtain complexity results.

\begin{prop}\label{aSKSNormKS}
If $\phi$ is the flow of a $\aSKS$ proof, with normal form $\psi$ under $\GNorm$, then $\psi$ is free of $\awu $ and $\acu $ nodes, i.e.\ contains just $\KS$ nodes.
\end{prop}
\begin{proof}
We argue by contradiction. Notice that by $\rcdcu$ we can assume that all $\acu$ nodes are above all $\acd$ nodes in $\psi$, by deformation, and so must have upper end directly connected to an $\awd$ or $\aid$ node, since $\phi$ is associated with a proof. However then $\psi$ can be reduced by either $\ridcu$ or $\rwdcu$, contradicting normality. An $\awu $ node, similarly, must be above all $\acd $ and $\acu $ nodes by $\rcdwu$ and $\rcuwu$, and so must have upper end directly connected to a $\awd $ or $\aid $ node, and again $\psi$ can be reduced by $\rwdwu$ or $\ridwu$, contradicting normality.
\end{proof}

Notice that the above proposition, along with previous results in this section, allows us to transform $\aSKS$ proofs to $\KS$ proofs of the same conclusion. We state this formally in Thm.~\ref{NormaSKStoKS} once we have determined more about the complexity of this transformation, which we address in the proceeding results.

\begin{defi}
An \emph{$\ai$-path} is a (simple) path that changes (vertical) direction only at $\aid$ and $\aiu$ nodes. We say that an $\ai  $-path is \emph{maximal} if it cannot be extended, and that it is \emph{open} if it begins and ends at a pending end of an edge.

The \emph{inversion} of a path is just the same path in the reverse direction.
\end{defi}

\begin{exa}\label{MaxAiPaths}
The paths on the right, and their inversions, are exactly all the maximal $\ai$-paths of the flow below. All of these are open except $0$, since one of its ends is a $\awd $ node. More generally all open paths are maximal but not vice-versa.
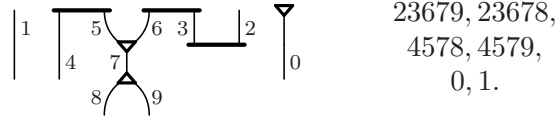
\afnegspace
\[
\begin{array}{c}
\\
\atomicflow{
(-2,4)*{\afvjd8{}{1}{}{}{}{}};
(4,8)*{\afaidnw{}{}{}{}{}{}};
(2,4)*{\afvju8{}{4}{}{}{}{}};
(12,8)*{\afaidnw{}{}{}{}{}{}};
(8,4)*{\afacd{5}{}{}{6}{7}{}};
(8,-2)*{\afacunw{8}{}{}{9}{}{}};
(16,4)*{\afaiu{3}{}{}{2}{}{}};
(22,8)*{\afawd{}{}{}{}{}{}};
(22,2)*{\afvjd4{}{0}{}{}{}};
}
\qquad\quad
\begin{array}{c}
  23679, 23678,  \\
  4578, 4579, \\
  0,1. \\
\end{array}
\end{array}
\]

\end{exa}



The following results allow us to estimate the size of the normal form of a flow, under $\GNorm$, without actually constructing it.

\begin{obs}\label{ConPath}
$\underset{\GNorm}{\longrightarrow}$ preserves the number of open $\ai$-paths in a flow.
\end{obs}


\begin{nota}
We write $\ulcorner\phi\urcorner$ to denote the number of open $\ai$-paths in a flow $\phi$, modulo inversions, and $\#(\rho,\phi)$ to denote the number of $\rho$ nodes in $\phi$ for an atomic structural rule $\rho$.
\end{nota}

\begin{lem}\label{PathNode}
If $\phi$ is the flow of a $\KS$ proof then $\ulcorner\phi\urcorner = \#(\aid,\phi)$.
\end{lem}
\begin{proof}
Since the flow of a proof can have no edge with upper end pending, every edge must be path-connected to a $\awd$ or $\aid$ node. Since no open path goes through an $\awd$ node, and there are no $\aiu$ nodes, every open $\ai$-path goes through a unique $\aid$ node, and every $\aid $ node accommodates at least one such path since there are no $\awu $ nodes. 

Now, the only other node a path can go through is an $\acd$ node. Consequently every node in an open $\ai$-path has in-degree 1 before passing its unique $\aid$ node, and out-degree 1 after. Hence each $\aid$ node accommodates exactly one open $\ai$-path.
\end{proof}

\begin{lem}\label{SizePath}
If $\phi$ is the flow of a $\aSKS$ proof, with normal form $\psi$ under $\GNorm$, then $|\psi|=O(|\phi|+\ulcorner\phi\urcorner)$.
\end{lem}
\begin{proof}
Let $\chi = \phi\normform_\weak$ so that $\phi\overset{*}{\underset{\weak}{\rightarrow}}\chi\overset{*}{\underset{\cont}{\rightarrow}}\psi$ by Lemma~\ref{WkContStrat}. By inspection of the rules of $\weak$ we must be able to decompose $\chi$ into two disjoint components: $\chi_1$, consisting of just $\aid $, $\acd $ and $\acu $ nodes, and $\chi_2$ consisting of just $\awd $ and $\awu $ nodes. In fact, since $\phi$ was the flow of a proof, $\chi_2$ consists of just $\awd $ nodes.

Now notice that any $\underset{\cont}{\rightarrow}$ derivation from $\chi$ to $\psi$ acts only on $\chi_1$, and so $\psi$ can be decomposed into two disjoint components: $\psi_1$, which is the normal form of $\chi_1$ under $\GNorm$, consisting of just $\aid $ and $\acd $ nodes by Prop.~\ref{aSKSNormKS}, and $\psi_2 = \chi_2$, consisting of just $\awd $ nodes.

We have that $|\psi_2| = |\chi_2|\leq|\chi|\leq|\phi|$ by Prop.~\ref{SizeBoundsWkContRed}.



Notice that $|\psi_1 | = 2\cdot\#(\aid,\psi_1 ) + \#(\acd,\psi_1 )$, since each $\aid$ node has two edges and each $\acd$ node adds a single extra edge. Since an $\acd $ node has in-degree $2$ and out-degree $1$, the number of $\acd$ nodes cannot outnumber the number of edges emanating from $\aid$ nodes in $\psi_1$, i.e.\ $\#(\acd,\psi_1) \leq 2\cdot\#(\aid,\psi_1)$, so we have $|\psi_1 |\leq 4\cdot\#(\aid,\psi_1 )$. By Lemma \ref{PathNode} we then have that $|\psi_1 |\leq 4\cdot\ulcorner\psi\urcorner$, and by Obs.\ \ref{ConPath} that $|\psi_1 |\leq 4\cdot\ulcorner\phi\urcorner$. 

Putting these together we obtain $|\psi| = |\psi_1 | + |\psi_2 | \leq |\phi| + 4\cdot\ulcorner\phi\urcorner$ as required.
\end{proof}

\begin{thm}\label{NormaSKStoKS}
If $\Phi$ is a $\aSKS$ proof with flow $\phi$ then we can transform it to a $\KS$ proof of the same conclusion in time polynomial in $|\Phi| + \ulcorner\phi\urcorner$.
\end{thm}
\begin{proof}
Let $\psi$ be the normal form of $\phi$ under $\GNorm$. By Prop.~\ref{aSKSNormKS} we have that $\psi$ contains just $\KS$ nodes, and so by Thm.~\ref{NormFormLiftsPoly} we have that a $\KS$ proof with same conclusion as $\Phi$ can be constructed in time polynomial in $|\Phi| + |\psi|$. The statement follows by the bound $|\psi| = O(|\phi| + \ulcorner\phi\urcorner ) = O(|\Phi| + \ulcorner\phi\urcorner)$ given by Lemma~\ref{SizePath}.
\end{proof}

\subsection{Length of atomic flows}
It is not difficult to see that the main contributor to an increase of flow size reducing under $\GNorm$ is the rule $\rcdcu$. It can sometimes cause an exponential blowup, as evident in Ex.~\ref{ExRedStrat} if there were no $\awd $ node at the top.



The following result provides a simple estimate of the number of paths in a flow, and also the complexity of flow normalisation under $\GNorm$, and is just an adaptation of well-known techniques to estimate the number of paths in any directed acyclic graph. 

\begin{defi}[Dimensions of a flow]
The \emph{length} of a flow is the maximum number of times the type of node changes in an $\ai$-path. The \emph{width} of a flow is the maximum size of a connected subflow containing just one type of node. The \emph{breadth} of a flow is the number of connected components it has.
\end{defi}

The above definition is perhaps most easily understood by allowing $\acd$ and $\acu$ nodes to have unbounded in-degree and out-degree respectively. For example, by just collapsing any configuration of $n-1$ connected $\acd$ nodes to a single node $\atomicflow{
(0,0)*{\afacd{}{}{}{}{}{}};
(0,4)*{\vldots};
(0,2)*{\text{\emph{\scriptsize{n}}}};	
}$, and similarly for configurations of $\acu$ nodes. Then we can consider the length of the flow to be essentially just the maximum length, in the usual sense, of an $\ai$-path, and the width is the maximum out-degree or in-degree.

It is worth mentioning here that replacing $\acd$ and $\acu$ nodes with these `super' nodes can be lifted polynomially to proofs, in the sense that $\acd$ and $\acu$ steps in a proof can be soundly replaced by similar `super' steps with only linear change in size.

\begin{prop}\label{CocConSizePath}
If $\phi$ is a flow consisting of just $\aSKS$ nodes with width $w$, length $l$ and breadth $b$, then $\ulcorner\phi\urcorner = b\cdot w^{\frac{l}{2}+O(1)}$.
\end{prop}
\begin{proof}
For simplicity we write $\atomicflow{
(0,0)*{\afacd{}{}{}{}{}{}};
(0,4)*{\vldots};
(0,2)*{\text{\emph{\scriptsize{n}}}};	
}$ for some configuration of $n-1$ connected $\gcdown$ nodes, and similarly for a $\gcu$ configuration. Notice that it suffices to consider the case when $b=1$, since paths in different connected components are disjoint.

In the worst case scenario we just have a sequence of $\atomicflow{
( 0,6)*{\afacd\one{}{}\two{}{}};
( 0,0)*{\afacunw\three{}{}\four};
(-4,0)*{\invisiblemark};
( 4,0)*{\invisiblemark};
(0,10)*{\vldots};
(0,8)*{\text{\emph{\scriptsize{w}}}};
(0,-1.5)*{\vldots};
(0,0)*{\text{\emph{\scriptsize{w}}}};
}$ configurations in series vertically, and each configuration multiplies the number of paths by $w$. It does not make much difference if there is a $\atomicflow{
(4,4)*{\afaidex{}{}{}{}{}{}42 };
(0,-2)*{\afacunw{}{}{}{}{}{}};
(8,-2)*{\afacunw{}{}{}{}{}{}};
(0,-2)*{\text{\emph{\scriptsize{w}}}};
(0,-3.5)*{\vldots};
(8,-2)*{\text{\emph{\scriptsize{w}}}};
(8,-3.5)*{\vldots};
}$ configuration at the top of the flow because of the way we have defined an $\ai$-path; the exponent differs by just the addition of a constant.
%
%
\end{proof}

\begin{rem}\label{BreadthWidthSize}
We generally use the trivial upper bound of size of flow for width and breadth, yielding the estimate $\ulcorner\phi\urcorner = |\phi|^{\frac{l}{2} + O(1)}$ for a $\aSKS$ flow $\phi$.
\end{rem}

\subsection{Contraction loops in atomic flows}
%
%
Sometimes the estimate of number of open $\ai  $-paths given via length of a flow in the previous section is not sufficiently accurate. An example of this is given in Sect.~\ref{TreeResSim} where the length of flows in the translation $R$ is polynomial in their size, and so the estimate is exponential, yet the actual number of open $\ai  $-paths is bounded above by a polynomial.

Unsurprisingly, it is only certain interactions between $\acd$ and $\acu$ nodes that generate significant complexity in a $\aSKS$ flow; in this section we focus on one type of interaction we call a `contraction loop'. In the absence of these the number of open $\ai$-paths in a flow is still polynomial in its size, regardless of its length.

\begin{defi}\label{Dfn:ConLoop}
A \emph{contraction loop} in a flow is a pair of $(\acu,\acd)$ nodes $(\nu_1 , \nu_2)$ such that there are at least two disjoint (directed) paths between $\nu_1$ and $\nu_2$
\end{defi}

For example we give the following flow and all its contraction loops,
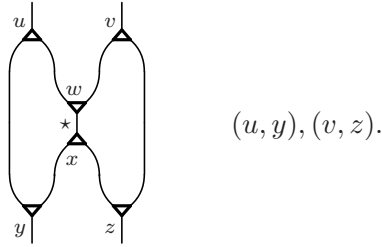
\[
\atomicflow{
(2,16)*{\afacu{}{}{}{}{u}{}};
(10,16)*{\afacu{}{}{}{}{v}{}};
(6,8)*{\afacd{}{w}{}{}{\star}{}};
(6,2)*{\afacunw{}{x}{}{}{}{}};
(2,-4)*{\afacd{}{}{}{}{y}{}};
(10,-4)*{\afacd{}{}{}{}{z}{}};
(0,9)*{\afvju6{}{}{}{}{}{}};
(0,3)*{\afvj6{}{}{}{}{}{}};
(12,9)*{\afvju6{}{}{}{}{}{}};
(12,3)*{\afvj6{}{}{}{}{}{}};
}
\qquad\quad
(u,y), (v,z).
\]
whereas every other pair has only one path between them. If the edge $\star$ were broken and there was no path from $w$ to $x$ then there would be no contraction loops at all.

\newcommand{\weight}{w}

\begin{lem}\label{NoConLoopsQuad}
If there are no contraction loops in a $\aSKS$ flow $\phi$ then $\ulcorner\phi\urcorner\leq|\phi|^{3}$.
\end{lem}
\begin{proof}
%
For an edge $\epsilon$ in $\phi$ consider the following two notions:
\begin{itemize}
\item  The \emph{weight} of $\epsilon$, denoted $\weight(\epsilon)$, is the number of directed paths from $\epsilon$ to the bottom of $\phi$, i.e.\ to a $\awu $ node or an edge with lower end pending.
\item For an atomic structural rule $\rho$ let $N(\rho,\epsilon)$ denote the number of $\rho$ nodes below $\epsilon$ that are connected to $\epsilon$ by a directed path.
\end{itemize}
We show that $\weight(\epsilon)\leq N(\awu ,\epsilon) + N(\acu,\epsilon) + 1$ by induction on the distance of $\epsilon$ from the bottom of $\phi$. The inequality is clear for the base cases when $\epsilon$ is directly connected to a $\awu $ node, $\ \aflower{\atomicflow{(0,0)*{\afawu{}{}{\epsilon}{}}}} \ $, and when $\epsilon$ has lower end pending, $\ \atomicflow{(0,0)*{\afvju{4}{\epsilon}{}}} \ $. We have two inductive steps:
\afnegspace
\begin{enumerate}
\item $\epsilon$ is an upper edge of a $\acd $ node,  $ \ \atomicflow{(0,0)*{\afacd{\epsilon}{}{}{}{\delta}{}}} \ $ . In this case we clearly have that $\weight(\epsilon) = \weight(\delta)$ and so the inequality follows by the inductive hypothesis.
\item $\epsilon$ is the upper edge of a $\acu $ node,  $\ \atomicflow{(0,0)*{\afacu{\gamma}{}{}{\delta}{\epsilon}{}}} \ $ . Observe that, since there are no contraction loops in $\phi$ and $\acd	 $ is the only node type with in-degree greater than $1$, any node below this $\acu $ node can be directed-path-connected to at most one of $\gamma$ or $\delta$. Consequently we have that $N(\awu,\epsilon) = N(\awu,\gamma)+N(\awu,\delta)$ and $N(\acu ,\epsilon) = N(\acu ,\delta)+N(\acu ,\gamma) + 1$. Therefore, 
\[
\begin{array}{ll}
w(\epsilon) &= w(\delta) + w(\gamma) \\
& \leq (N(\awu ,\delta) + N(\acu ,\delta) + 1) + (N(\awu ,\gamma) + N(\acu ,\gamma) + 1) \\
& \leq (N(\awu ,\delta)+N(\awu ,\gamma) + (N(\acu ,\delta)+N(\acu ,\gamma) + 1) +1 \\
& \leq N(\awu ,\epsilon) + N(\acu ,\epsilon) + 1
\end{array}
\]
\end{enumerate}
Finally, notice that $|\phi|\geq N(\awu ,\phi)+N(\acu ,\phi)+1$ and so we have shown that $w(\epsilon)\leq |\phi|$.

Clearly the number of open $\ai  $-paths going through an edge with upper end pending is bounded above by its weight, and so by $|\phi|$ by the bound above, while the number of open $\ai  $-paths going through any $\aid $ node is bounded above by the product of the weights of each of its edges, and so by $|\phi|^2$. In particular we have that the number of open $\ai  $-paths going through any edge at the top of a flow is bounded above by $|\phi|^2$, and there are at most $|\phi|$ many such edges, whence the bound follows.

\end{proof}

\begin{exa}
In fact the bound given above is optimal, up to multiplication by a constant. Consider the flow $\phi$ below, where there are $n$ $\aid $ nodes:
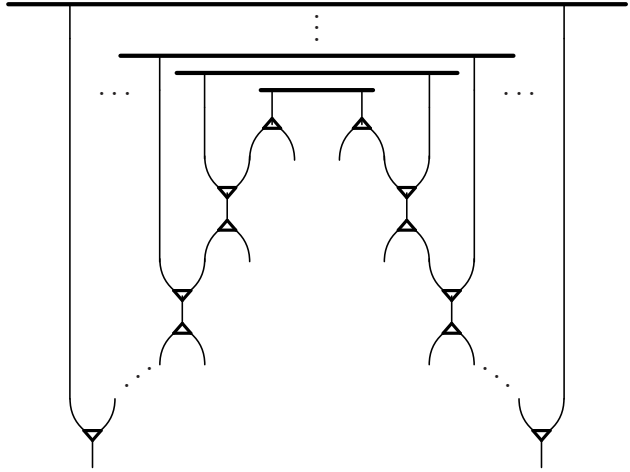
\[
\atomicflow{
(-4,-6)*{\afacunw{}{}{}{}};
(-8,-12)*{\afacdnw{}{}{}{}{}{}};
(-8,-16)*{\afacu{}{}{}{}{}{}};
(-12,-24)*{\afacdnw{}{}{}{}};
(-12,-28)*{\afacu{}{}{}{}{}{}};
(-20,-40)*{\afacd{}{}{}{}{}{}};
(-10,-5)*{\afvj{6}};
(-14,-10)*{\afvj{20}};
(-22,-15)*{\afvj{42}};
(4,-6)*{\afacunw{}{}{}{}};
(8,-12)*{\afacdnw{}{}{}{}{}{}};
(8,-16)*{\afacu{}{}{}{}{}{}};
(12,-24)*{\afacdnw{}{}{}{}};
(12,-28)*{\afacu{}{}{}{}{}{}};
(20,-40)*{\afacd{}{}{}{}{}{}};
(10,-5)*{\afvj{6}};
(14,-10)*{\afvj{20}};
(22,-15)*{\afvj{42}};
(0,0)*{\afaidex{}{}{}{}{}{}{4}{2}};
(0,2)*{\afaidex{}{}{}{}{}{}{10}{2}};
(0,4)*{\afaidex{}{}{}{}{}{}{14}{2}};
(0,8)*{\vlvdots};
(0,10)*{\afaidex{}{}{}{}{}{}{22}{2}};
(18,0)*{\dots};
(-18,0)*{\dots};
(-16,-32)*{\iddots};
(16,-32)*{\ddots};
}
\]
Clearly the flow has size linear in $n$ and notice that the weights (as defined in the above proof) of the topmost edges, starting from the left, are $1,2,\dots,n,n,\dots,1$ respectively. Consequently the number of open $\ai  $-paths going through the $\aid $ nodes, starting from the outside, is $1^2, 2^2, \dots,n^2$ respectively, and so taking the sum we obtain $\ulcorner\phi \urcorner=\Omega(n^3)$.

Notice also that $\phi$ has length linear in $n$ and width $2$, and so the upper bound on $\ulcorner\phi\urcorner$ given by Prop.~\ref{CocConSizePath} is exponential, considerably worse than the bound given in Lemma~\ref{NoConLoopsQuad}.
\end{exa}

\section{Truth tables and tree-like Gentzen systems}\label{DITreeTT}
$\KS$ polynomially simulates tree-like cut-free Gentzen systems since its rules are generalisations of Gentzen rules \cite{BrunTiu:01:A-Local-:ly} \cite{BrusGugl:07:On-the-P:fk}. In the other direction Bruscoli and Guglielmi have proved in \cite{BrusGugl:07:On-the-P:fk} that the converse does not hold, by way of the so-called `Statman tautologies'. We offer a new proof here, via truth tables, as an exercise prior to the main results.

Let $LK^-$ denote the cut-free sequent calculus, and tree-$LK^-$ denote the system of tree-like proofs in this calculus.\footnote{It does not matter too much which version of the sequent calculus we choose as complexity properties are typically preserved across common variations.}

\begin{prop}[D'Agostino]\label{NSimTT}
Tree-$LK^-$ cannot polynomially simulate truth tables.
\end{prop}
\begin{proof}
See \cite{springerlink:10.1007/BF00156916}.
\end{proof}

To expand slightly on the above proposition, truth tables are efficient when there are exponentially many occurrences of each atom, and some such tautologies are hard for tree-$LK^-$. One such example, used by D'Agostino, is simply the disjunction of every assignment on $k$ propositional variables, what we call $\bigvee_\mathcal{A} \gamma_\mathcal{A}$ below. Such a formula has size $k\cdot 2^k$, although any tree-like cut-free sequent proof must contain at least $k!$ branches, while a truth table contains $2^k$ rows and $k\cdot 2^k$ columns.



\begin{lem}
$\aSKS$ polynomially simulates truth tables.
\end{lem}

\begin{proof}
Let $\tau$ be a tautology. For each partial assignment $\mathcal{A}$, defined on just those variables appearing in $\tau$, and each formula $A$ satisfied by $\mathcal{A}$ construct a derivation $\Phi_{\mathcal{A}}(A)$ by structural induction on $A$ as follows:
\[
\Phi_{\mathcal{A}}(a)\equiv a
\quad,\quad
\Phi_{\mathcal{A}}(B\vlan C) \equiv \Phi_{\mathcal{A}}(B)\vlan\Phi_{\mathcal{A}}(C)
\quad,\quad
\Phi_{\mathcal{A}}(B\vlor C) \equiv \vlinf{=}{}{\vls[B . \vlinf{\gwd}{}{C}{\fff} ]}{\Phi_{\mathcal{A}}(B)}
\]
where, in the last case, when $A$ is a disjunction, the disjunct $B$ was chosen such that $B$ is satisfied by $\mathcal{A}$. It is clear that each $\Phi_{\mathcal{A}}(\tau)$ has conclusion $\tau$ and premiss a conjunction of literals; moreover this conjunction of literals is satisfied by $\mathcal{A}$. 

Let $\gamma_\mathcal{A}$ be the conjunction of all literals satisfied by $\mathcal{A}$, so that each variable appears exactly once. Then we can easily construct derivations $\vlderivation{
\vldd{}{\{\awu,\acu\}}{\prem(\Phi_{\mathcal{A}}(\tau))}{\vlhy{\gamma_{\mathcal{A}}}}}$.

Now construct a proof $\Psi$ of $\bigvee_{\mathcal{A}} \gamma_{\mathcal{A}}$ in $\{\aid ,\acu,\swi,\med\}$ by induction on the number of distinct variables, as shown below:
\[
\text{\emph{Base case}:}\quad \vlinf{\aid}{}{\vls[a.\bar{a}]}{\ttt} 
\quad,\quad
\text{\emph{Inductive step}:}\quad
\vlinf{2\cdot\swi}{}{\bigvee_{\mathcal{A}}{(a\vlan\gamma_{\mathcal{A}})} \vlor \bigvee_{\mathcal{A}}{(\bar{a} \vlan\gamma_{\mathcal{A}})}}{
\vlproofd{\Psi}{\{\aid ,\acu,\swi,\med\}}{\bigvee_{\mathcal{A}}{\left(\vlinf{=}{}{\vls(\vlinf{\aid}{}{\vls[a.\bar{a}]}{\ttt}.\vlinf{\gcu}{}{\vls(\gamma_{\mathcal{A}}.\gamma_{\mathcal{A}})}{\gamma_{\mathcal{A}}})}{\gamma_{\mathcal{A}}}\right)}}
}
\]
Finally we put these together and apply contractions to obtain a $\aSKS$ proof of $\tau$:
\[
\vlderivation{
\vldd{}{\{\gcdown\}}{\tau}{
\vlpd{\Psi}{\{\aid ,\acu,\swi,\med\}}{\bigvee_{\mathcal{A}}{\left(\vlderivation{
\vldd{\Phi_{\mathcal{A}}}{\{\gwd\}}{\tau}{
\vldd{}{\{\awu,\acu\}}{\prem(\Phi_{\mathcal{A}} (\tau))}{\vlhy{\gamma_{\mathcal{A}}}}
}
}\right)}
}
}
}
\]
It is clear that the derivations inside the large parentheses have size polynomial in $|\tau|$, which is the number of columns in a truth table, and the number of these derivations appearing in disjunction is just the number of assignments, which is the number of rows in a truth table.
\end{proof}

\begin{thm}\label{SimTT}
$\KS$ polynomially simulates truth tables.
\end{thm}

\begin{proof}
Notice that, in the above simulation, all $\acu$ steps are above all $\acd$ steps, and so the associated flow will have bounded length. The result follows by Cor.~\ref{FlowTransProofs}, Lemma~\ref{SizePath} and Prop.~\ref{CocConSizePath}.
\end{proof}

\begin{cor}
Tree-$LK^-$ cannot polynomially simulate $\KS$.
\end{cor}

\begin{proof}
Immediate from Prop.~\ref{NSimTT} and Thm.~\ref{SimTT}.
\end{proof}
\begin{rem}
It should be noted that D'Agostino's separation is only quasipolynomial,\footnote{A quasipolynomial is a function of size $n^{\log^{\Theta(1)} n}$. By taking logarithms and using the integral bound it is not difficult to see that $k!$ grows quasipolynomially in $2^k$.} and so our separation is also only quasipolynomial, while the proof using the Statman tautologies in \cite{BrusGugl:07:On-the-P:fk} yields an exponential separation. Nonetheless, an exponential separation follows from the results in the next section.
\end{rem}

\section{Separations via variants of the pigeonhole principle}\label{StrongViaFPHP}
Je{\v r}\'abek has shown that $\aSKS$ has polynomial-size proofs of the functional and onto variants of the pigeonhole principle \cite{Jera::On-the-C:kx}. We show that these proofs reduce under $\GNorm$ to $\KS$ proofs with only polynomial increase in size, and so cut-free sequent calculi, Resolution and bounded-depth Frege systems are unable to polynomially simulate $\KS$. 





The pigeonhole principle states that, if $n$ pigeons sit in $n-1$ holes, some hole must contain more than one pigeon. In its unrestricted formulation, the mapping from pigeons to holes can be many-many, while in the functional variant it must be many-one, i.e.\ a function, and the onto variant insists that each hole must be occupied. It is clear that the two variants are weaker than the unrestricted version, and the variant containing both criteria, the onto functional pigeonhole principle, is weaker still. We will see this more clearly in the following definition.


The propositional encodings of pigeonhole principle variants below are most easily understood by interpreting the atoms $a_{ij}$ as ``pigeon $i$ sits in hole $j$''. Recall that it does not matter how large disjunctions and conjunctions are bracketed, since any valid bracketing can be easily reduced to any other by the $=$ rule.

\newcommand{\PHP}{\mathsf{PHP}}
\newcommand{\FPHP}{\mathsf{F}\PHP}
\newcommand{\OPHP}{\mathsf{O}\PHP}
\newcommand{\OFPHP}{\mathsf{OF}\PHP}
\newcommand{\LFPHP}{\mathsf{L}\FPHP}
\newcommand{\RFPHP}{\mathsf{R}\FPHP}
\begin{defi}[Pigeonhole principles]
We define the following formulae,
\[
\PHP_{n} \equiv \bigvee\limits^{n}_{i=0}{\bigwedge\limits^{n}_{j=1}{\bar{a}_{ij}}} \vlor \bigvee\limits^{n}_{i=0} \bigvee\limits_{ j<j' } (a_{ij} \vlan a_{ij'})
,\quad
\mathsf{F}_n \equiv  \bigvee\limits^{n}_{j=1} \bigvee\limits_{i<i'} {(a_{ij}\vlan a_{i'j})}
,\quad
\mathsf{O}_n  \equiv  \bigvee\limits^n_{j=1}\bigwedge\limits^n_{i=0} \bar{a}_{ij}
\]
and denote by $\FPHP_n$, $\OPHP_n$ and $\OFPHP_n$ the formulae obtained by putting in disjunction the associated formulae, i.e.\ $\FPHP_n \equiv \mathsf{F}_n \vlor \PHP_n$, $\OPHP_n \equiv \mathsf{O}_n \vlor \PHP_n$ and $\OFPHP_n \equiv \mathsf{O}_n \vlor \mathsf{F}_n \vlor \PHP_n$.
\end{defi}

We can see in the above definition that any variant can be obtained from a stronger variant, i.e.\ one with a subset of disjuncts, by a simple application of generic weakening $\gwd$. Consequently upper bounds on the size of proofs of one variant yield upper bounds for all weaker variants, and lower bounds vice-versa.


The following result was proved by Beame, Impagliazzo and Pitassi, and independently by Kraj{\'\i}{\v c}ek, Pudl{\'a}k and Woods.

\begin{thm}\label{BdFNfPHP}
Bounded-depth Frege systems have only exponential-size proofs of $\OFPHP_n$.
\end{thm}

\begin{proof}
See \cite{KraPudWood:95:ExpPHPbdF}, \cite{PitBeaImp:93:ExpLBPHP}.
\end{proof}

\begin{cor}\label{GentzNfPHP}
Bounded-depth Frege systems, Resolution and cut-free sequent calculi have only exponential-size proofs of all variants of the pigeonhole principle.
\end{cor}
\begin{proof}
All the systems are just special cases of bounded-depth Frege, and a proof of any variant can be extended to one of $\OFPHP_n$ by an application of generic weakening $\gwd$.
\end{proof}

On the other hand we have the following:
\begin{thm}[Buss]\label{PolyPHPFrege}
There are polynomial-size Frege proofs of $\PHP_n$, and so all variants of the pigeonhole principle.
\end{thm}

\begin{proof}
See \cite{Buss:87:Polynomi:fk}.
\end{proof}



From here, polynomial-size $\SKS$ proofs of $\PHP_n$ follow from the following observation:
\begin{prop}[Bruscoli and Guglielmi]\label{SKSFrege}
$\SKS$ is polynomially equivalent to Frege systems.
\end{prop}
\begin{proof}
See \cite{BrusGugl:07:On-the-P:fk}.
\end{proof}


Notice that one direction of the above proposition, indeed the direction that we require, that $\SKS$ polynomially simulates Frege systems, can be obtained by recognising that the rules of $\SKS$ are just generalisations of the rules of Gentzen systems with cut, which are well-known to be polynomially equivalent to Frege systems.


The following trick, now standard in the deep inference literature, is very useful for proving certain tautologies in $\KS$, as we will see. The idea is that if we know there is an $\SKS$ proof of a tautology, then we can transform it into a $\KS$ proof of that tautology in disjunction with trivial contradictions. If we can find derivations from each of these contradictions to the tautology we want to prove, then the composition mechanisms for derivations in deep inference guarantee that we can build a proof of the tautology.

\begin{lem}\label{SwitchCutTrick}
Let $A$ be a formula over the atoms $a_1 , \dots , a_n$. Every $\SKS$ proof $\Phi$ of $A$ can be polynomially transformed to a $\KS$ proof of $A\vlor\bigvee_{i} (a_i \vlan\bar{a}_i )$.
\end{lem}

\begin{proof}
See e.g.\ \cite{Jera::On-the-C:kx}, \cite{BrusGuglGundPari:09:Quasipol:kx}.
\end{proof}

\begin{lem}[Je{\v r}\'abek]\label{JerFPHPProofs}
There are polynomial-size proofs of $\FPHP_{n}$ and $\OPHP_n$ in $\aSKS$.
\end{lem}
\begin{proof}
By Thm.~\ref{PolyPHPFrege}, Prop.~\ref{SKSFrege} and Lemma \ref{SwitchCutTrick} there are polynomial-size $\KS$ proofs with conclusion $\vls[\PHP_{n}.\bigvee_{i,j} (a_{ij}\vlan\bar{a}_{ij})]$. For each atom $a_{st}$ we construct a derivations $\Phi^{a_{st}}_n $ in $\aSKS\setminus\{\acd\}$ with premiss $a_{st}\vlan\bar{a}_{st}$ and conclusion $\FPHP_n$ as follows:
\[
\vlderivation{
\vlin{\gwd}{}{\FPHP_n }{
\vlin{2\cdot\swi}{}{\bigwedge\limits_j \bar{a}_{sj} \vlor \left(\vlder{}{\{\swi\}}{\bigvee\limits_{j\neq t} \vls(a_{st} .a_{sj} )}{\vls(\vlinf{(n-2)\cdot\acu}{}{a_{st} \vlan \vldots a_{st}}{a_{st} } . \bigvee\limits_{j\neq t } a_{sj} )}\right)}{
\vlin{=}{}{\vls(a_{st} . \bar{a}_{st} . \vlinf{\gidown}{}{\vls[\bigwedge\limits_{j\neq t} \bar{a}_{sj} . \bigvee\limits_{j\neq t} a_{sj} ]}{\ttt})}{\vlhy{\vls(a_{st}.\bar{a}_{st})}}
}
}
}
\]
We then put these together and apply contractions to obtain proofs of $\FPHP_n$:
\[
\vlderivation{
\vldd{}{\{\gcdown\}}{\FPHP_{n}}{
\vlpd{}{\KS}{\left[ \PHP_n \vlor \bigvee\limits_{i,j } \left( \vlderd{\Phi_{n}^{a_{ij}} }{\aSKS\setminus\{\acd\}}{\FPHP_n }{a\vlan\bar{a}} \right) \right]}
}
}
\]
We can construct similar derivations from $a_{st}\vlan \bar{a}_{st}$ to $\OPHP_n$, given below, and put them together in a similar way to obtain proofs of $\OPHP_n$ in $\aSKS$.
\[
\vlderivation{
\vlin{\gwd}{}{\OPHP_n}{
\vlin{2\cdot\swi}{}{\bigwedge\limits_i \bar{a}_{it} \vlor \left(\vlderivation{
\vlde{}{\{\swi\}}{\bigvee\limits_{i\neq s}\vls(a_{st} . a_{it})}{\vlhy{\vls(\vlinf{(n-2)\cdot\acu}{}{a_{st} \vlan \vldots \vlan a_{st}}{a_{st}} . \bigvee\limits_{i\neq s} a_{it})}}
}\right)}{
\vlin{=}{}{\vls(a_{st} . \bar{a}_{st} . \vlinf{\gidown}{}{\bigwedge\limits_{i\neq s} \bar{a}_{it}  \vlor \bigvee\limits_{i\neq s} a_{it}}{\ttt})}{\vlhy{\vls(a_{st}. \bar{a}_{st})}}
}
}
}
\]
\afnegspace
\end{proof}


\begin{thm}\label{KSfPHP}
There are polynomial-size $\KS$ proofs of $\FPHP_{n}$, $\OPHP_n$ and so also $\OFPHP_n$.
\end{thm}
\begin{proof}
In the proofs of $\FPHP_n$ constructed above, Lemma~\ref{JerFPHPProofs}, notice that the only $\acu$ steps occur in $\Phi^{a_{st}}_n$ where there are also no $\acd$ steps, and similarly for $\OPHP_n$. It follows that there are only two alternations between $\acd$ and $\acu$ steps in the path of any atom from an $\aid$ step, and so the atomic flows of these proofs will have bounded length. The result follows by Thm.~\ref{NormaSKStoKS} and Prop.~\ref{CocConSizePath}.
\end{proof}

\begin{cor}\label{GenResBDFNoP}
Cut-free sequent calculi, Resolution and bounded-depth Frege systems are exponentially separated from $\KS$.
\end{cor}
\begin{proof}
Immediate from Cor.~\ref{GentzNfPHP} and Thm.~\ref{KSfPHP}.
\end{proof}

\section{Polynomial simulations of versions of Resolution}\label{VerResSim}
In this section we present a polynomial simulation in $\KS$ of tree-like and multiset Resolution systems. We point out that, in previous work \cite{Das:12:Complexi:kx}, a more general result was claimed, namely a simulation of unrestricted Resolution operating with sets, however that proof contained errors that we could not amend for this version of the article. The status of those results is unresolved.

We define the Resolution system below in both its set and multiset formulations. 

\begin{defi}[Resolution]\label{Dfn:Resolution}
We use symbols $\Gamma, \Delta$ etc.\ to range over (multi)sets of literals and write `$\cup$' to denote (multi)set union.
We define the system \emph{(mulitset-)Resolution} by the following rules:
\[
\vlinf{\weak}{}{\Gamma\cup\Delta}{\Gamma}
\quad , \quad
\vliinf{\res}{}{\Gamma\cup\Delta}{\Gamma\cup\{a\}}{\Delta\cup\{\bar{a}\}}
\]
A \emph{derivation} from (multi)sets $\Gamma_1 , \dots , \Gamma_s $ is a list\footnote{We do not mind how the list is delimited. In the translations below they are written vertically for presentation reasons.} $ \pi = (\Delta_1 , \dots , \Delta_n)$ where each $\Delta_i$ is some $\Gamma_j$ or follows by one of the rules above whose premisses have occurred previously in the list. If we further have that $\Delta_n = \emptyset$ then we call $\pi$ a \emph{refutation} of $\Gamma_1 , \dots , \Gamma_s$.

We call a derivation $(\Delta_1 , \dots , \Delta_n)$ \emph{tree-like} if each $\Delta_i$ is used at most once as the premiss of any inference step concluding some $\Delta_j$.
%
\end{defi}

To simplify the treatment of (multiset-)Resolution derivations we address certain rather pathological situations below. These are also the reason why we opt to include weakening in our formulation.\footnote{It is straightforward to show that the formulation without weakening is polynomially equivalent by a rule permutation argument.}

\begin{rem}[Assumptions on derivation format]
In a (multiset-)Resolution derivation, we can assume that neither premiss of a $\res$ step contains both a resolved atom and its dual. Any such step would have the following format,
\[
\vliinf{\res}{}{\Gamma\cup\Delta\cup\{a\}}{ \Gamma\cup\{a\} }{\Delta\cup\{a , \bar a \}}
\]
which can be simulated by an application of $\weak$ to the left premiss, if necessary at all.

We consequently have that no $\res$ step has identical premisses, and so there is no need to consider this case in the translations that follow.

Finally, in the case of sets, we further assume that neither the resolved atom nor its dual appear in the conclusion of a $\res$ step. I.e., in the definition of the $\res$ above, we assume that $\Gamma$ and $\Delta$ contain neither $a$ nor $\bar{a}$. Aside from the above situation this phenomenon might also occur if, say, $a\in\Gamma$ and so $\Gamma = \Gamma\cup \{a \}$, whence the step may be simulated by an application of $\weak$ in a similar way.
\end{rem}

Before presenting our simulations in $\KS$ of versions of Resolution, we set some notational conventions below in order to easily switch between the two settings.

\begin{nota}[Set symbols in deep inference]\label{SetSymbDI}
To reduce the amount of syntax in our deep inference derivations, we will simply write $\Gamma$ instead of $\bigvee \Gamma$ for the disjunction of the members of $\Gamma$, as an abuse of notation.

We similarly use other (multi)set-theoretic notation. In particular, in light of the above remark, we will always have that $\Gamma\cup \{a\} = \Gamma \vlor a$ in both the set and multiset settings whenever it occurs. In the multiset setting we further have that $\Gamma\cup\Delta = \Gamma \vlor \Delta$ for all $\Gamma, \Delta$.
\end{nota}


\begin{rem}
Throughout this section, when we say that a proof system polynomially simulates a refutation system, we mean that every refutation of $A$ in the latter can be polynomially transformed to a proof of $\bar A$ in the former.
\end{rem}

\begin{defi}[Dual systems]
The dual of a deep inference rule $\vlinf{\rho}{}{B}{A}$ is $\vlinf{\bar\rho}{}{\bar{A}}{\bar{B}}$. E.g.\ in $\SKS$ a rule $x\downarrow$ is dual to $x\uparrow$, while $\swi$ and $\med$ are self-dual. The set of duals of a system $\mathcal{S}$ is denoted $\overline{\mathcal{S}}$. Notice that a rule is sound if and only if its dual is, by the law of contraposition.

The dual of a derivation $\vlder{\Phi}{\mathcal S}{B}{A}$ is the derivation $\vlder{\bar\Phi}{\bar{\mathcal S}}{\bar A}{\bar B}$ obtained by flipping $\Phi$ upside down and replacing each inference step with its dual.

Similarly we define the dual of a flow-rewriting rule to be the rule flipped upside-down, replacing each node with its dual, and the dual of a rewriting system is just the set of its duals. A rule/system is sound, terminating and/or confluent if and only if its dual is.
\end{defi}
%

We are now ready to define our basic translation from Resolution derivations to deep inference, on which our simulation results in later sections will be based.

\begin{defi}[Translation of Resolution derivations]\label{TransResSKS}
We give the following translation $R$ of a Resolution step to a $\overline{\aSKS}$ derivation,
\[
\begin{array}{rrcl}
R: &\quad \vlinf{\weak}{}{\Gamma\cup\Delta}{\Gamma} \quad & \mapsto & \quad \vlderivation{
\vlin{=}{}{ \Gamma \cup \Delta}{
\vlin{=}{}{  \Gamma \vlor \vlinf{\gwd}{}{ \Delta \setminus \Gamma }{\fff} }{ \vlhy{ \Gamma} }
}
}
\\
\noalign{\bigskip}
R:& \quad \vliinf{\res}{}{\Gamma\cup\Delta}{\Gamma\cup\{a\}}{\Delta\cup\{\bar{a}\}} \quad & \mapsto & \quad 
\vlderivation{
\vliq{(\gcdown }{\scriptstyle{)}}{  \Gamma\cup \Delta }{
\vlin{=}{}{ \Gamma \vlor  \Delta  }{
\vlin{2\cdot\swi}{}{  \Gamma \vlor  \Delta \vlor \vlinf{\aiu}{}{\fff}{ a \vlan \bar a  }}{\vlhy{
\vlinf{=}{}{\Gamma\vlor a}{\Gamma\cup\{ a \}} \vlan \vlinf{=}{}{\Delta \vlor \bar a}{\Delta \cup\{\bar a\}}
}}
}
}
}
\end{array}
\]
where the parenthesised $\gcdown $ steps are present only when $\Gamma$ and $\Delta $ are sets (not multisets) that have nonempty intersection.

We extend the definition of $R$ to any (multiset-)Resolution derivation $\pi = (\Delta_1 , \dots , \Delta_n)$ from $\Gamma_1 , \dots , \Gamma_s$, such that $R$ has the following format:
\[
R \quad :\quad \begin{array}{c}
\Delta_1 \\
\vdots \\
\Delta_n
\end{array}
\quad\mapsto\quad
\vlderd{}{\overline{\aSKS}}{\bigwedge\limits_{m=1}^n \Delta_m  }{\bigwedge\limits_{r=1}^s \Gamma_r }
\]
The definition is by induction on the length $n$ of the Resolution derivation $\pi$.

If $n=0$, i.e.\ $\pi$ is an empty list, then $\bigwedge_m \Delta_m = \ttt $ and we define:\footnote{Throughout this translation we shall omit certain $=$ rules required to handle units, e.g.\ we associate the empty disjunction with $\fff$ and the empty conjunction with $\ttt$. By now this is a routine consideration.}
\[
R: \quad\pi \quad \mapsto\quad \vlinf{\gwu}{}{\ttt}{\bigwedge\limits_{r=1}^s \Gamma_r }
\]

If $\pi$ is extended by a (multi)set $\Gamma_i$ then we define,
	\[
	R:\quad
	\begin{array}{c}
	\Delta_1 \\
	\vdots \\
	\Delta_n \\
	\Gamma_i
	\end{array}
	\quad\mapsto\quad
	\vlderivation{
	\vlin{=}{}{ \vlderd{\it{ID}}{}{\bigwedge\limits^{n}_{m=1}  \Delta_m  }{ \bigwedge\limits^s_{r=1}  \Gamma_r } \vlan  \Gamma_i }{
	\vlin{=}{}{ \bigwedge\limits_{r \neq i}  \Gamma_r  \vlan \vlinf{\gcu}{}{ \Gamma_i \vlan  \Gamma_i }{\Gamma_i}}{\vlhy{\bigwedge\limits^s_{r=1}  \Gamma_r}  }
	}
	}
	\]
where the derivation marked \emph{ID} is already defined by induction.

If $\pi$ is extended by a $\weak$ step then we define,
	\[
	R:\quad
	\begin{array}{c}
	\Delta_1 \\
	\vdots \\
	\Delta_i \\
	\vdots \\
	\Delta_n \\
	\Delta_i \cup \Sigma
	\end{array}
	\quad\mapsto\quad
	\vlderivation{
			\vlin{=}{}{ \bigwedge\limits^{n+1}_{m=1}  \Delta_m }{
			\vlin{=}{}{ \bigwedge\limits_{m \neq i}  \Delta_m \vlan \vlinf{\gcu}{}{  \Delta_i  \vlan \vlderd{R(\weak)}{}{  \Delta_i \cup \Sigma }{ \Delta_i  } }{  \Delta_i \ } }{
			\vldd{\it{ID}}{}{ \bigwedge\limits^n_{m=1}  \Delta_m }{ \vlhy{\bigwedge\limits^s_{r=1}  \Gamma_r} }
			}
			}
			}
	\]
where $\Delta_{n+1}= \Delta_i \cup \Sigma$ and the derivation marked \emph{ID} is already defined by induction.

If $\pi$ is extended by a $\res$ step then we define,
	\[
	R: \quad\begin{array}{c}
			\Delta_1 \\
			\vdots \\
			\Delta_i' \cup \{a\} \\
			\vdots \\
			\Delta_j' \cup \{\bar a\} \\
			\vdots \\
			\Delta_n \\
			\Delta_i' \cup \Delta_j'
		\end{array}
		\quad\mapsto\quad
		\vlderivation{
		\vlin{=}{}{ \bigwedge\limits^{n+1}_{m=1}  \Delta_m }{
		\vlin{=}{}{ \bigwedge\limits_{m \neq i,j}  \Delta_k \vlan \vlinf{=}{}{  \Delta_i \vlan  \Delta_j \vlan \vlderd{R(\res)}{}{  \Delta_i' \cup \Delta_j' }{ \Delta_i \vlan  \Delta_j } }{  \vlinf{\gcu}{}{\Delta_i \vlan \Delta_i}{\Delta_i} \vlan \vlinf{\gcu}{}{\Delta_j \vlan \Delta_j}{\Delta_j } } }{
		\vldd{\it{ID}}{}{ \bigwedge\limits^n_{m=1}  \Delta_m }{ \vlhy{\bigwedge\limits^s_{r=1}  \Gamma_r} }
		}
		}
		}
	\]
where $\Delta_i = \Delta_i' \cup\{a \}$, $\Delta_j = \Delta_j' \cup\{\bar{a}\}$, $\Delta_{n+1} = \Delta_i' \cup \Delta_j'$ and the derivation marked \emph{ID} is already defined by induction.
\end{defi}

%

\subsection{Multiset Resolution}
We consider the multiset setting, for arbitrary derivations (i.e.\ not necessarily tree-like). Here we achieve a polynomial simulation in $\KS$ rather easily by using the identities pointed out in Not.~\ref{SetSymbDI}.

\begin{obs}\label{ConFreeMultiRes}
The image of a multiset-Resolution derivation under $R$ has no $\acd $ steps.
\end{obs}
\begin{proof}
The only $\acd $ steps in the definition of $R$ occur in the translation of individual $\res$ steps, where the format of the premiss is $\Gamma \vlor \Delta$ and the conclusion $\Gamma \cup \Delta$. However, as already pointed out in Not.~\ref{SetSymbDI}, these are equivalent in the multiset setting, and so no $\acd $ steps are necessary at all.
\end{proof}

\begin{thm}\label{KSSimMultRes}
$\KS$ polynomially simulates multiset resolution.
\end{thm}
\begin{proof}
For a multiset resolution refutation $\pi$ of $\Gamma_1, \dots , \Gamma_s$ we have that $R\pi$ is a $\overline{\aSKS}$ derivation from $\bigwedge\limits^s_{r=1} \Gamma_r$ to $\fff$, by definition, not containing $\acd $ steps, by the observation above. Consequently we have that $\overline{R\pi}$ has the following format,
\[
\vlderd{\overline{R\pi}}{\KS\cup\{\awu \}}{\bigvee\limits^s_{r=1} \bar \Gamma_r}{\ttt}
\]
whose flow has bounded length (due to the absence of $\acu $ nodes) and so can be transformed to a $\KS$ proof of the same conclusion in polynomial time by Prop.~\ref{CocConSizePath} and Thm.~\ref{NormaSKStoKS}.
\end{proof}

\subsection{Tree-like Resolution}\label{TreeResSim}
We now consider Resolution derivations over sets and show a polynomial simulation of tree-like refutations in $\KS$. Here the argument is only slightly more involved: there are both $\acd $ and $\acu $ steps occurring, but they do not interact in any complex way.

We point out that the proof of this result could perhaps be carried out more simply by writing tree-like derivations as trees and simulating steps locally, discarding premisses once they are used. However the current presentation allows us to use the same translation $R$ for this and the multiset setting, as well as the extensions introduced in the next section. This uniform treatment is possible due to Lemma~\ref{NoConLoopsQuad} on the number of open $\ai  $-paths in the absence of contraction loops.

%
%

\begin{prop}\label{PendingEdgeTreeResTrans}
If $\pi$ is a tree-like Resolution derivation then every $\acu $ node in $\flow(R\pi)$ has one edge pending.
\end{prop}
\begin{proof}
By induction on the length of $\pi$. On inspection of the atomic flows, notice that each translation step introduces only $\gcu $ nodes whose lower left edge is pending (with respect to the horizontal order of atoms given in the definition of $R$ in Dfn.~\ref{TransResSKS}), so it suffices to show that new nodes are not attached to the lower edge of existing $\gcu $ nodes.

If any such situation existed then, by construction, each $\gcu $ node must be associated with the same set $\Delta_i$ that is the premiss of distinct steps in $\pi$, contradicting the fact that $\pi$ is tree-like.
%
%
%
%
%
\end{proof}

Recall the notion of contraction loop from Dfn.~\ref{Dfn:ConLoop}.

\begin{cor}\label{NoConLoopsTreeResTran}
If $\pi$ is a tree-like Resolution derivation then $\flow(R\pi)$ has no contraction loops. 
\end{cor}
\begin{proof}
Any contraction loop would violate Prop.~\ref{PendingEdgeTreeResTrans} above.
\end{proof}

\begin{thm}
	$\KS$ polynomially simulates tree-like Resolution.
	\label{KSSimTreeRes}
\end{thm}
\begin{proof}
For a tree-like Resolution refutation $\pi$ of $\Gamma_1 , \dots , \Gamma_s$, as in the proof of Thm.~\ref{KSSimMultRes}, we have that $\overline{R\pi}$ has the format:
\[
\vlderd{\overline{R\pi}}{\aSKS}{\bigvee\limits^s_{r=1} \bar \Gamma_r}{\ttt}
\]
We also have that $\flow(\overline{R\pi})$ contains no contraction loops, since otherwise there would also be a contraction loop in $\flow(R\pi)$, violating Cor.~\ref{NoConLoopsTreeResTran}. The result now follows from Lemma~\ref{NoConLoopsQuad} and Thm.~\ref{NormaSKStoKS}.
\end{proof}

\subsection{Extensions of Resolution}
Finally, we notice that this simulation extends to some extensions of Resolution, operating on (multi)sets of conjunctions of literals, introduced by Kraj{\'\i}{\v c}ek in \cite{Krajicek01onthe}. These systems are known to be strictly stronger than their counterparts we have so far dealt with \cite{Segerlind02aswitching} \cite{Esteban04onthe}.

\newcommand{\Res}{\text{Resolution}}
\begin{defi}\label{Dfn:Resf}
Let $s,t,$ etc.\ denote terms, i.e.\ conjunctions of literals, and $\Gamma,\Delta,$ etc.\ now denote sets of terms. We define the following rules,
\[
\vlinf{\weak}{}{\Gamma\cup\Delta}{\Gamma}
\quad , \quad
\vliinf{\vlan}{}{\Gamma\cup \Delta \cup\{s\vlan t \}}{\Gamma\cup\{s\} }{\Delta\cup\{t\}}
\quad,\quad
\vliinf{\res}{}{\Gamma\cup\Delta}{\Gamma\cup\{s\} }{\Delta\cup\{\bar a_1,\dots,\bar a_k \} }
\]
where in the $\res$ rule we require that $s=\bigwedge\limits^k_{i=1} a_i$.

For a function $f:\mathbb{N}\to\mathbb{N}$, a (tree-like) (multiset-)$\Res(f)$ derivation/refutation is defined analogously to Dfn.~\ref{Dfn:Resolution}, using the rules above, with the additional proviso that no term has size bigger than $f(N)$, where $N$ is the number of inference steps in the refutation.
\end{defi}

\begin{thm}
For any $f:\mathbb{N}\to\mathbb{N}$ we have the following:
\begin{enumerate}
\item\label{KSSimExtMultResItem} $\KS$ polynomially simulates multiset-$\Res(f)$.
\item\label{KSSimExtTreeResItem} $\KS$ polynomially simulates tree-like $\Res(f)$.
\end{enumerate}
\end{thm}
\begin{proof}[Proof sketch]
The proofs are analogous to those appearing earlier for the tree-like and multiset variants of Resolution respectively, interpreting $\Gamma, \Delta$ etc.\ as (multi)sets of terms and replacing variables $a,\bar a$ etc.\ by term variables and their duals $t , \bar t $ etc.\footnote{In light of the aforementioned abuse of notation identifying a (multi)set with the disjunction of its elements, we associate the dual of term $\bigwedge\limits^k_{i=1} a_i$ with the (multi)set $\{\bar a_1 , \dots , \bar a_k  \} $.}

We extend the definition of $R$ on individual steps to $\wedge$ steps as follows:
\[
R:\quad
\vliinf{\vlan}{}{\Gamma\cup\Delta\cup\{s\vlan t \}}{\Gamma\cup \{s\} }{\Delta \cup\{t\} }
\quad\mapsto\quad
\vlderivation{
\vliq{(\gcdown }{\scriptstyle{)} }{ \Gamma\cup\Delta\cup\{ s\vlan t \} }{
\vlin{2\cdot\swi}{}{\Gamma\vlor \Delta \vlor (s\vlan t)}{ \vlhy{
\vlinf{=}{}{\Gamma\vlor s}{\Gamma\cup\{ s \}} \vlan \vlinf{=}{}{\Delta \vlor t}{\Delta \cup\{t\}}
} }
}
}
\]
$R$ is extended to $\Res(f)$-derivations similarly to Dfn.~\ref{TransResSKS}, dealing with $\wedge$ steps in the same way $\res$ steps.

For \ref{KSSimExtMultResItem} the argument is identical to that for Thm.~\ref{KSSimMultRes} since, again, $\gcdown $ steps are not introduced for the same reason, cf.\ Not.~\ref{SetSymbDI}.

For \ref{KSSimExtTreeResItem} the argument is identical to that for Thm.~\ref{KSSimTreeRes} since the argument of Prop.~\ref{PendingEdgeTreeResTrans} remains valid and so there are still no contraction loops in the flow of a derivation in the image of $R$.
\end{proof}

%

\section{Conclusions}
We have presented a series of upper bound results for the deep inference system $\KS$. Polynomial simulations were given for truth tables and versions of Resolution, while polynomial-size proofs of functional and onto variants of the propositional pigeonhole principle were presented, yielding exponential separations from all these systems, as well as bounded-depth Frege systems.

We have seen that atomic flows can act as a powerful tool to analyse and manipulate derivations, and that often we can avoid the possibly exponential blowup arising from the $\rcdcu$ rule. A relevant pursuit would be to investigate whether we can \emph{always} avoid this blowup via, perhaps, a polynomial-time local or global flow reduction; this would imply that $\KS$ polynomially simulates $\aSKS$, and thus quasipolynomially simulates $\SKS$ and Frege systems. Despite much work in this direction a result, positive or negative, seems difficult to obtain; this remains arguably the most important question in the proof complexity of deep inference.

\begin{qu}
Does $\KS$ (quasi)polynomially simulate $\aSKS$?
\end{qu} 

This question has already been asked in previous works, e.g.\ \cite{BrusGugl:07:On-the-P:fk}, \cite{Jera::On-the-C:kx}, \cite{Stra:08:Extensio:kk} and \cite{Das:11:Depth-Change}, with both positive and negative answers conjectured. We point out that our contribution might be helpful to any work towards a positive answer.

We also point out that it could be that atomic flows do not themselves include sufficient information to carry out an efficient normalisation procedure from $\aSKS$ to $\KS$, if one exists at all. To this end there is work ongoing by various researchers to augment atomic flows with certain \emph{logical} information, in the hope of accessing further normalisation procedures and the ability to `rebuild' deep inference proofs efficiently from flow objects.\footnote{This is known to be impossible in polynomial time, by \cite{DBLP:conf/rta/Das13}, under certain hardness assumptions.}

One might argue that our proofs in Sects.~\ref{DITreeTT} and \ref{StrongViaFPHP} eventually reduced to flows of bounded length, and so the complexity of normalisation was trivially polynomial. A more sophisticated situation might involve flows of bounded width and logarithmic length, again resulting in a polynomial blowup, or quasipolynomial width and polylogarithmic length, giving a quasipolynomial blowup. We refer the reader to the recent article \cite{Das:13:The-Pige:fk} for an example of this, where quasipolynomial-size proofs of the unrestricted pigeonhole principle are given in $\KS$, utilising the techniques of this paper.

All the results that appear in this work, and indeed all other works on proof complexity of deep inference, e.g.\ \cite{BrusGugl:07:On-the-P:fk} \cite{Jera::On-the-C:kx} \cite{Das:13:The-Pige:fk}, present only simulations and upper bounds. We currently know of no technique for proving lower bounds for deep inference systems and, in light of the question above and results in \cite{Jera::On-the-C:kx} and \cite{Das:13:The-Pige:fk}, it may be that such an endeavour, even for $\KS$, could be as difficult as that of proving lower bounds for Frege systems.

\section*{Acknowledgements}
I am grateful to my advisor, Alessio Guglielmi, for many fruitful discussions on this and related work, as well as Tom Gundersen, Lutz Stra\ss burger and the anonymous reviewers for this paper and previous versions of it. I owe particular thanks to Arnold Beckmann who spotted crucial errors in previous versions and gave detailed feedback.

\bibliographystyle{alpha}
\bibliography{biblio}

\end{document}